\documentclass[journal,12pt,onecolumn]{IEEEtran}

\usepackage{amsmath,amsfonts}
\usepackage{amssymb,verbatim,amsfonts,amsmath,amsthm}
\usepackage{geometry}
\usepackage{cite}
\usepackage{hyperref}
\usepackage{bm}

\usepackage{longtable}
\usepackage{booktabs}
\usepackage[justification=centering]{caption}
\usepackage{graphicx}
\usepackage[figuresright]{rotating}

\usepackage{makecell}

\newtheorem{theorem}{Theorem}
\newtheorem{lemma}[theorem]{Lemma}
\newtheorem{corollary}[theorem]{Corollary}

\newtheorem{example}[theorem]{Example}
\newtheorem{remark}[theorem]{Remark}

\begin{document}

\title{Linear Complementary Pairs of Algebraic Geometry Codes via Kummer Extensions}
\author{Junjie Huang\thanks{J. Huang, H. Chen and H. Zhang are with the School of Mathematics, Sun Yat-sen University, Guangzhou 510275, China (e-mail: huangjj76@mail2.sysu.edu.cn; e-mail: chenhj69@mail2.sysu.edu.cn; e-mail: zhanghch56@mail2.sysu.edu.cn).}, Haojie Chen, Huachao Zhang,
Chang-An Zhao$^*$\thanks{C.-A. Zhao is with the School of Mathematics, Sun Yat-sen University, Guangzhou 510275, China, and also with the Guangdong Key Laboratory of Information Security Technology, Guangzhou 510006, China (e-mail: zhaochan3@mail.sysu.edu.cn).\\
*-corresponding author}}

\maketitle

\begin{abstract}

Due to their widespread applications, linear complementary pairs (LCPs) have attracted much attention in recent years. In this paper, we determine explicit construction of non-special divisors of degree $g$ and $g-1$ on Kummer extensions with specific properties. In addition, we present several methods for constructing LCPs of algebraic geometry codes (AG Codes) via Kummer extensions. These results are applied in constructing explicit LCPs of AG Codes from subcovers of the BM curve, elliptic function fields, hyperelliptic function fields and other function fields. It is important to mention that we construct several families LCPs of MDS AG Codes from elliptic function fields and we obtain some linear complementary dual (LCD) codes from certain maximal elliptic function fields and hyperelliptic function fields.

\end{abstract}

\section{Introduction}

Let $\mathbb{F}_q$ denote the finite field with $q$ elements. A $q$-ary $[n,k,d]_q$ linear code $\mathcal{C}$ is a linear subspace of $\mathbb{F}_q^n$ with dimension $k$ and minimum distance $d$. A pair $(\mathcal{C},\mathcal{D})$ of linear codes of length $n$ over $\mathbb{F}_q$ is called a linear complementary pair (LCP) \cite{carletLinearComplementaryPairs2018} if their direct sum $\mathcal{C}\oplus\mathcal{D}$ equals to the whole space $\mathbb{F}_q^n$ (i.e. the two codes have trivial intersection and complementary dimensions). Due to their widespread applications, linear complementary pairs (LCPs) have attracted much attention in recent years. In symmetric cryptography, Direct Sum Masking (DSM) leverages the properties of LCPs to enhance security against the side-channel attacks (SCA) and the fault injection attacks (FIA) (see\cite{Orthogonal_Direct_Sum_Masking},\cite{Side-Channel_Attacks},\cite{Linear_complementary_dual_code_improvement_to_strengthen_encoded_circui_tagainst_hardware_Trojan_horses}). In this context, the minimum distance $d(\mathcal{C})$ measures the protection against FIA, while the minimum distance $d(\mathcal{D}^\perp)$ of the dual of the code $\mathcal{D}$ assesses the defense against SCA. It is shown that the level of resistance against both FIA and SCA is determined by $\min(d(\mathcal{C}),d(\mathcal{D}^\perp))$, known as the security parameter of the LCP of the codes. Note that the linear complementary dual (LCD) codes represent a specific case when $\mathcal{D} = \mathcal{C}^\perp$, in which case the security parameter is simply the minimum distance of $\mathcal{C}$. Furthermore, in the field of coding theory, LCPs and LCD codes can be employed to construct entanglement-assisted quantum error-correcting codes (EAQECCs) through the application of Theorem~4 presented in \cite{galindoEntanglementassistedQuantumErrorcorrecting2019}.

In recent years, there has been a lot of research on LCD codes (see\cite{ComplementaryDualAGC},\cite{ARAYA2019270},\cite{arayaMinimumWeightsBinary2020},\cite{arayaCharacterizationClassificationOptimal2021},\cite{ARAYA2021101925},\cite{bouyuklievaOptimalBinaryLCD2021},\cite{doi:10.1142/S0129054119500242},\cite{galvezBoundsBinaryLCD2018},\cite{haradaBinaryLinearComplementary2019},\cite{haradaConstructionBinaryLCD2021},\cite{KeitaIshizuka2023AdvancesinMathematicsCommunications},\cite{7797200},\cite{liSeveralConstructionsOptimal2024}), but there is relatively little research specifically focused on the LCPs of codes. Following the introduction of the concept of LCPs by Bhasin et~al. in~\cite{Linear_complementary_dual_code_improvement_to_strengthen_encoded_circui_tagainst_hardware_Trojan_horses}, Carlet et~al. subsequently investigated constacyclic and quasi-cyclic LCPs of codes in \cite{carletLinearComplementaryPairs2018}. Further, Borello et~al. gave a short and elementary proof of the fact that for a LCP $(\mathcal{C},\mathcal{D})$, the code $\mathcal{D}$ is uniquely determined by $\mathcal{C}$ and the dual code $\mathcal{D}^\perp$ is permutation equivalent to $\mathcal{C}$, where $\mathcal{C}$ and $\mathcal{D}$ are $2$-sided ideals in a group algebra~\cite{BORELLO2020111905}; Hu et al. proved a necessary and sufficient condition for a pairs of linear codes
over finite rings to be LCPs \cite{huLinearComplementaryPairs2021}; and Bhowmick et al. provided a necessary and sufficient condition for a pair of codes $(\mathcal{C},\mathcal{D})$ to be LCP of codes over finite non-commutative Frobenius rings~\cite{bhowmickLinearComplementaryPairs2024a}. Additionally, in 2023, Choi et al. presented infinite families of binary LCPs which are optimal in the sense that their security parameters reach the highest minimum distance $d_L$ \cite{choiOptimalBinaryLinear2023}, and Güneri proved that for every $k\ge5$ and $d\ge \lceil(k-1)/2\rceil 2^{k-1}$ there exist binary LCPs of codes of length $g(k,d)$ in the same year~\cite{10156864}. 

Algebraic geometry codes (AG Codes) have revealed their potential in numerous areas of coding theory. In 2018, Mesnager et al. conducted pioneering research on algebraic geometry LCD codes over finite fields and provided concrete examples from algebraic curves \cite{ComplementaryDualAGC}. Further, Beelen et al. 
explicitly constructed several families of MDS LCD AG Codes over rational function field \cite{8319441}. However, it was not until 2024 that the first systematic investigation into LCPs of AG Codes took place, with Bhowmick, Kumar Dalai, and Mesnager making significant contributions in this domain \cite{On_linear_complementary_pairs_of_algebraic_geometry_codes}. Bhowmick et al. pioneeringly established the characterization and construction mechanism of LCPs of AG Codes over finite fields and provided illustrative examples to demonstrate their findings. The sufficient conditions for constructing LCPs of AG Codes largely rely on the existence of non-special divisors of degree $g-1$ in the function field. In~\cite{Explicit_Non-special_Divisors_of_Small_Degree}, Moreno et al. presented an explicit construction of non-special divisors of small degree over Kummer extensions. These divisors support the construction of LCD AG Codes on some families of Kummer extensions. Subsequently, Castellanos et al. in \cite{Linear_Complementary_Dual_codes_and_Linear_Complementary_Pairs_of_AG_codes_in_function_fields} focused on the effective construction of LCPs of AG Codes and LCD AG Codes over function fields of genus $g\ge1$. Note that, in \cite{Explicit_Non-special_Divisors_of_Small_Degree} and \cite{Linear_Complementary_Dual_codes_and_Linear_Complementary_Pairs_of_AG_codes_in_function_fields}, the authors considered certain maximal function fields of particular interest in coding theory, including the Hermitian function field, elliptic function fields and so on. More generally, they considered Kummer extensions given by
\begin{equation}\label{kummer_normal}
    y^m = \prod_{i=1}^r(x - \alpha_i),
\end{equation}
over $\mathbb{F}_q$ with $(q,m) = 1$ and $(r,m)=1$. In \cite{Linear_Complementary_Dual_codes_and_Linear_Complementary_Pairs_of_AG_codes_in_function_fields}, Castellanos et al. established pairs of suitable divisors that yielding non-special divisors of degree $g-1$ in the function field defined by Equation (\ref{kummer_normal}) with $r<m$ and they constructed LCPs of AG Codes and LCD AG Codes in Kummer extensions, hyperelliptic function fields, and elliptic curves. 

In this paper, we present several methods for constructing LCPs of AG Codes via Kummer extensions with specific properties. More precisely, we consider Kummer extensions given by
\begin{equation}\label{kummer_non-normal}
    y^m = \prod_{i=1}^r(x - \alpha_i)^{\lambda_i},
\end{equation}
over $\mathbb{F}_q$ with $(q,m) = 1$ and some other properties. To achieve our goals, we determine an explicit construction of non-special divisors of degree $g$ and $g-1$ on these Kummer extensions. We build upon the seminal work of Ballet and Le Brigand \cite{BALLET2006293}, who demonstrated the existence of non-special divisors of degree $g-1$ (respectively,  $ g $) for function fields of genus $g\ge2$ over finite fields with order $q\ge4$ (respectively, $q\ge3$). In addition, our constructions are also applicable to the case of Kummer extensions defined by Equation (\ref{kummer_normal}). That is to say, we can also obtain the same LCPs of AG Codes as in \cite{Linear_Complementary_Dual_codes_and_Linear_Complementary_Pairs_of_AG_codes_in_function_fields}. Besides, our results are applied in constructing explicit LCPs of AG Codes from subcovers of the BM curve, elliptic function fields, hyperelliptic function fields and other function fields. Our method of constructing LCPs of AG Codes from elliptic function fields is completely different from that in \cite{Linear_Complementary_Dual_codes_and_Linear_Complementary_Pairs_of_AG_codes_in_function_fields}, and we also construct several families LCPs of MDS AG Codes from elliptic function fields which are optimal in the sense that their security parameters reach the Singleton upper bound. 

The paper is organized as follows. Section~\ref{pre} provides some preliminaries about the theory of function fields, Weierstrass semigroups, and the concept of LCPs of AG Codes. In the subsequent section, we present an explicit construction of non-special divisors of degree $g$ on Kummer extensions with specific properties (see Theorem \ref{non-special_divisors_cons}). Following that, we give several examples (see Example \ref{non-special_divisors_cons_ex}) about constructing non-special divisors of degree $g$ and $g-1$ on Kummer extensions defined by Equation (\ref{kummer_non-normal}). In Section~\ref{gen_con_LCP}, we provide several construction of LCPs of AG Codes over Kummer extensions (see Theorems \ref{con_w=0_1}, \ref{con_w=0_2}, \ref{con_v=0} and Remark \ref{con_w=0_rem}). In Section~\ref{LCPs_of_AG Codes}, we specifically construct LCPs of AG Codes from subcovers of the BM curve, elliptic function fields, hyperelliptic function fields and other function fields. In addition, we obtain two families of LCD codes from certain maximal elliptic function fields and hyperelliptic function fields. At last, we give a summary of the entire article in Section \ref{conclusion}. 


\section{Preliminaries}\label{pre}

In this section, we briefly recall some preliminaries on function fields, Weierstrass semigroups, algebraic geometry codes and linear complementary pairs of algebraic geometry codes. 

\subsection{Function Fields and Weierstrass Semigroups}
Let $F/\mathbb{F}_q$ be a function field of genus $g$ with the full constant field $\mathbb{F}_q$. Let $\mathbb{P}_F$ denote the set of places of $F$. The discrete valuation associated with $P\in\mathbb{P}_F$ is denoted by $v_P$. The divisor group of $F/\mathbb{F}_q$ is defined as the free abelian group which is generated by $\mathbb{P}_F$; it is denoted by $\text{Div}(F)$. Assume that $D=\sum_{P\in\mathbb{P}_F}n_PP$ is a divisor such that almost all $n_P=0$. Then the degree of $D$ is defined by $\deg D=\sum_{P\in\mathbb{P}_F}n_P\deg P$ and the support of $D$ by ${\rm supp}(D)=\{P\in \mathbb{P}_F\mid n_P\neq0\}$. For $z\in F$, the principal divisor of $f$ is denoted by $(z)$, the zero divisor of $z$ is denoted by $(z)_0$ and the pole divisor of $z$ is denoted by $(z)_\infty$. Two divisors $D,D^\prime\in\text{Div}(F)$ are said to be linearly equivalent, written $D\sim D^\prime$, if $D=D^\prime+(z)$ for some $z\in F\backslash\{0\}$. 

The Riemann-Roch space associated to the divisor $D\in\text{Div}(F)$ is the finite-dimensional $\mathbb{F}_q$-vector space
\[
    \mathcal{L}(D)=\{z\in F\mid (z)+D\ge 0 \}\cup\{0\}.
\]
The dimension of $\mathcal{L}(D)$ is given by $\ell(D)$ and satisfies the Riemann-Roch Theorem \cite[Thm. 1.5.15]{Algebraic_Function_Fields_and_Codes}, i.e. 
\[
    \ell(D)=\deg D+1-g+\ell(K-D),
\]
where $K$ is a canonical divisor of $F$. Recall that the index of specialty of a divisor $D$ is 
\[
    i(D)=\ell(D)-(\deg D+1-g)=\ell(K-D).
\]
A divisor $D$ is called non-special if and only if $i(D)=0$, i.e. $\ell(D)=\deg D+1-g$. Moreover, if $\deg D>2g-2$, then $D$ is non-special. 

Now we introduce the notions concerning Weierstrass semigroups. Given $l$ distinct rational places of $F$, named $Q_1,\cdots,Q_l$, the Weierstrass semigroup is defined by
\[
    H(Q_1,\cdots,Q_l)=\bigg\{\alpha=(\alpha_1,\cdots,\alpha_l)\in\mathbb{N}^l\;\bigg| \;\exists\, f\in F\;\text{with}\;(f)_\infty=\sum_{i=1}^l\alpha_iQ_i\bigg\},
\]
and the Weierstrass gap set is defined by 
\[
    G(Q_1,\cdots,Q_l)=\mathbb{N}^l\,\backslash \,H(Q_1,\cdots,Q_l).
\]
Let $[n]:=\{1,\cdots,n\}$. To delineate the minimal generating set of Weierstrass semigroups, a more comprehensive set of symbols must be introduced. For $i\in[l]$, let $\Gamma^+(Q_i)=H(Q_i)$, and for $k\ge2$, let
\[
    \Gamma^+(Q_{i_1},\cdots,Q_{i_k})=\bigg\{ v\in\mathbb{N}^k\;\bigg|\;v\text{ is minimal in }\{w\in H(Q_{i_1},\cdots,Q_{i_k})\mid w_i=v_i\}\text{ for some }i\in[l]\bigg\}.
\]
For each subset $I\subseteq[l]$, let $\iota_I$ denote the natural inclusion $\mathbb{N}^{|I|}\hookrightarrow \mathbb{N}^l$ into the coordinates indexed by $I$. The minimal generating set of $H(Q_1,\cdots,Q_l)$ is defined by
\[
    \Gamma(Q_1,\cdots,Q_l)=\bigcup_{k=1}^l\bigcup_{I=\{i_1,\cdots,i_k\}\subseteq[l]}\iota_I\big(\Gamma^+(Q_{i_1},\cdots,Q_{i_k})\big). 
\]
Then the Weierstrass semigroup $H(Q_1,\cdots,Q_l)$ is completely determined by the minimal generating set $\Gamma(Q_1,\cdots,Q_l)$\cite[Thm. 7]{The_Weierstrass_Semigroup_of_an_m-tuple_of_Collinear_Points_on_a_Hermitian_Curve}:
\[
    H(Q_1,\cdots,Q_l) = \bigg\{ {\rm lub}\{u_1,\cdots,u_l\}\;\bigg|\;u_1,\cdots,u_l\in \Gamma(Q_1,\cdots,Q_l) \bigg\},
\]
where $u_i = (u_{i_1},\cdots,u_{i_l})$ and ${\rm lub}\{u_1,\cdots,u_l\}$ represents the least upper bound of vectors $u_i$ for all $1\le i\le l$, defined by
\[
    {\rm lub}\{u_1,\cdots,u_l\}=\big(\max\{u_{1_1},\cdots,u_{l_1}\},\cdots,\max\{u_{1_l},\cdots,u_{l_l}\}\big).
\]
The following lemma illustrates the relationship between the minimal generating set and non-special effective divisors of degree $g$. Before we provide the lemma, we introduce the partial order $\le$ on $\mathbb{N}^l$ which is defined by
\[
    v=(v_i)\le w=(w_i) \text{ if and only if } v_i\le w_i \text{ for all } i\in [l].
\]
Further we denote
\[
    v\nleq w \text{ if and only if } v_i> w_i \text{ for some } i\in [l].
\]
\begin{lemma}\cite[Prop. 5]{Explicit_Non-special_Divisors_of_Small_Degree}\label{non-special}
    Consider an effective divisor $A=\sum_{i=1}^l\alpha_iQ_i\in {\rm Div}(F)$ of degree $g$. If $\gamma\nleq \alpha$ for all $\gamma\in\Gamma(Q_1,\cdots,Q_l)$, then $A$ is non-special.
\end{lemma}

The subsequent lemma is employed to derive non-special divisors of degree $g-1$. The divisors of degree $g-1$ play an importance role in the constructions of linear complementary pairs, as elaborated in Section \ref{gen_con_LCP}.

\begin{lemma}\cite[Lem. 3]{BALLET2006293}\label{non_special_g-1}
    If $A$ is a non-special divisor of degree $g$ on a function field $F$ and there exists a rational place $P\in\mathbb{P}_F\,\backslash\,{\rm supp}(A)$, then $A-P$ is non-special. 
\end{lemma}

\subsection{Algebraic Geometry Codes}
For more details of algebraic geometry codes\,(AG Codes), the reader may refer to \cite{Algebraic-Geometric_Codes}. Let $F/\mathbb{F}_q$ be a function field of genus $g$ with the full constant field $\mathbb{F}_q$. Let $P_1,\cdots,P_n\in\mathbb{P}_F$ be $n$ pairwise distinct rational places of $F$ and $D=\sum_{i=1}^nP_i$. For a divisor $G$ of $F/\mathbb{F}_q$ with $2g-2<\deg G<n$ and ${\rm{supp}}(G)\cap {\rm{supp}}(D)=\varnothing$, the algebraic geometry code\,(AG Code) associated with the divisors $D$ and $G$ is defined as
\[
    \mathcal{C}_{\mathcal{L}}(D,G)=\{(x(P_1),\,\cdots,\,x(P_n))\mid x\in\mathcal{L}(G)\}\subseteq\mathbb{F}_q^n,
\]
where $\mathcal{L}(G)$ is the Riemann-Roch space with the dimension $\ell(G)=\deg G+1-g$. Then the AG Code $\mathcal{C}_{\mathcal{L}}(D,G)$ is an $[n,k,d]_q$ linear code with dimension $k=\ell(G)$ and minimum distance $d\ge n-\deg G$. On the other hand, for $\mathcal{C}_{\mathcal{L}}(D,G)$, the Singleton bound indicates that 
\begin{equation}\label{singleton_bound}
    d\le n-k+1=n-\deg G+g. 
\end{equation}
If the equality in (\ref{singleton_bound}) holds, then $\mathcal{C}_{\mathcal{L}}(D,G)$ is called maximum distance separable, or MDS. 

The following lemma gives the dual code of the AG Code $\mathcal{C}_{\mathcal{L}}(D,G)$. 
\begin{lemma}\cite[Prop. 2.2.10 and Prop. 8.1.2]{Algebraic_Function_Fields_and_Codes}\label{dual_AG Code}
    Let $t$ be an element of $F$ such that $v_{P_i}(t)=1$ for all $1\le i\le n$. Then the following assertions hold:
    \begin{itemize}
        \item The Weil differential $\eta=dt/t$ satisfies $v_{P_i}(\eta)=-1$ and $\eta_{P_i}(1)=1$ for $1\le i\le n$. 
        \item $\mathcal{C}_{\mathcal{L}}(D,G)^{\bot}=\mathcal{C}_{\mathcal{L}}(D,D-G+(\eta))=\mathcal{C}_{\mathcal{L}}(D,D-G+(dt)-(t))$. 
    \end{itemize}
\end{lemma}

Two codes $\mathcal{C}_1,\mathcal{C}_2\subseteq\mathbb{F}_q^n$ are said to be equivalent if there exists a vector ${\bf a}=(a_1,\cdots,a_n)\in(\mathbb{F}_q^\times)^n$ such that $\mathcal{C}_2 = {\bf a}\cdot \mathcal{C}_1$. If two AG Codes $\mathcal{C}_{\mathcal{L}}(D,G_1)$ and $\mathcal{C}_{\mathcal{L}}(D,G_2)$ satisfying that
\[
    G_1\sim G_2 \;\text{ and }\;{\rm{supp}}(G_1)\cap {\rm{supp}}(D)={\rm{supp}}(G_2)\cap {\rm{supp}}(D)=\varnothing,
\]
then $\mathcal{C}_{\mathcal{L}}(D,G_1)\,$ and $\,\mathcal{C}_{\mathcal{L}}(D,G_2)$ are equivalent, denoted by $\mathcal{C}_{\mathcal{L}}(D,G_1)\sim \mathcal{C}_{\mathcal{L}}(D,G_2)$. It is evident that equivalent codes have the same dimension and minimum distance.

\subsection{Linear Complementary Pairs of AG Codes}

A pair of linear codes $(\mathcal{C},\mathcal{D})$ over the finite field $\mathbb{F}_q$ with length $n$ is defined as a linear complementary pair (LCP) if it satisfies the condition $\mathcal{C}\oplus\mathcal{D}=\mathbb{F}_q^n$, where $\oplus$ denotes the direct sum of two subspaces. It is important to observe that if $\mathcal{D}=\mathcal{C}^\bot$, then $\mathcal{C}$ is referred to as a linear complementary dual (LCD) code. 

Let $F/\mathbb{F}_q$ be a function field of genus $g$ with the full constant field $\mathbb{F}_q$. Let $P_1,\cdots,P_n\in\mathbb{P}_F$ be $n$ pairwise distinct rational places of $F$ and $D=\sum_{i=1}^nP_i$. For two divisors $G,H$ with 
\[
    {\rm{supp}}(G)\cap {\rm{supp}}(D)={\rm{supp}}(H)\cap {\rm{supp}}(D)=\varnothing,
\]
the pair $(\mathcal{C}_{\mathcal{L}}(D,G),\mathcal{C}_{\mathcal{L}}(D,H))$ is an LCP of AG Codes over $\mathbb{F}_q$ if $\mathcal{C}_{\mathcal{L}}(D,G)\oplus\mathcal{C}_{\mathcal{L}}(D,H)=\mathbb{F}_q^n$. In other words, a pair $(\mathcal{C}_{\mathcal{L}}(D,G),\mathcal{C}_{\mathcal{L}}(D,H))$ is an LCP of AG Codes if and only if 
\[
    \dim (\mathcal{C}_{\mathcal{L}}(D,G))+\dim (\mathcal{C}_{\mathcal{L}}(D,H)) = n\; \text{ and }\; \mathcal{C}_{\mathcal{L}}(D,G)\cap\mathcal{C}_{\mathcal{L}}(D,H)=\{0\}. 
\]

For two divisors $G,H\in\text{Div}(F)$, we define 
\[
    \gcd(G,H) = \sum_{P\in\mathbb{P}_F}\min\{v_P(G),v_P(H)\}P\;\text{ and }\;{\rm lmd}(G,H) = \sum_{P\in\mathbb{P}_F}\max\{v_P(G),v_P(H)\}P. 
\]
In \cite{On_linear_complementary_pairs_of_algebraic_geometry_codes}, Bhowmick et al. established the sufficient conditions concerning the divisors $G$ and $H$ to derive an LCP of AG Codes.

\begin{lemma}\cite[Thm. 3.5]{On_linear_complementary_pairs_of_algebraic_geometry_codes}\label{LCP_criterion}
    Let $\mathcal{C}_{\mathcal{L}}(D,G)$ and $\mathcal{C}_{\mathcal{L}}(D,H)$ be two AG Codes, of length $n$, over a function field $F/\mathbb{F}_q$ of genus $g\neq 0$, such that $\ell (G) + \ell (H) = n $ and $2g-2<\deg G, \deg H<n$. Then the pair $(\mathcal{C}_{\mathcal{L}}(D,G),\mathcal{C}_{\mathcal{L}}(D,H))$ is an LCP if $\gcd(G,H)$ is a non-special divisor of degree $g-1$ and ${\rm lmd}(G,H)-D$ is a non-special divisor. 
\end{lemma}

In the preceding lemma, it is evident that non-special divisors play a crucial role in deriving LCPs of AG Codes. Consequently, in the following section, we will present a construction of non-special divisors derived from Kummer extensions.

\section{Construction of Non-Special Divisors from Kummer Extensions}\label{construction}

In this section, we will detail the construction of non-special divisors from Kummer extensions. We focus on Kummer extensions $F/\mathbb{F}_q(x)$ defined by
\begin{equation}\label{kum_diff}
    y^m=a\cdot\prod_{i=1}^u(x-\alpha_i)\cdot\prod_{j=1}^v(x-\beta_j)^{\lambda_j}\cdot\prod_{k=1}^w(x-\gamma_k)^{\lambda_k^\prime},\quad(q,m)=1\;\;\text{and}\;\;1<\lambda_j,\lambda_k^\prime<m,
\end{equation}
where $\alpha_i,\,\beta_j,\,\gamma_k\in\mathbb{F}_q \,(1\le i\le u,\,1\le j\le v,\,1\le k\le w)$ are pairwise distinct elements, $a\in\mathbb{F}_q\backslash\{0\}$ and 
\[
    \lambda_0=u+\sum_{j=1}^v\lambda_j+\sum_{k=1}^w\lambda_k^\prime,\quad(\lambda_0,m)=1,\quad(\lambda_j,m)=1,\quad (\lambda_k^\prime,m)\neq1.
\]
By \cite[Prop. 3.7.3]{Algebraic_Function_Fields_and_Codes}, the function field $F/\mathbb{F}_q$ has genus
\[
    g=\frac{(m-1)(u+v-1)+mw-\sum_{k=1}^w(\lambda_k^\prime,m)}{2}.
\]
Let $Q_i\,(1\le i\le u)$ be the (only) rational place in $\mathbb{P}_F$ corresponding to the zero of $x-\alpha_i$, $Q_j^\prime\,(1\le j\le v)$ the (only) rational place corresponding to the zero of $x-\beta_j$ and $Q_\infty$ the unique place at infinity of $F$. For $1\le j\le v$, let $1\le \lambda\le m-1$ be the inverse of $\lambda_j$ modulo $m$ and 
\[
    f_j(t) = \sum_{i=1}^u\bigg\lceil\frac{t\lambda}{m}\bigg\rceil +\sum_{j=1}^v\bigg\lceil\frac{t\lambda\lambda_j}{m}\bigg\rceil+\sum_{k=1}^w\bigg\lceil\frac{t\lambda\lambda_k^\prime}{m}\bigg\rceil-\bigg\lceil\frac{t\lambda\lambda_0}{m}\bigg\rceil - 1. 
\]
For convenience, we set
\[
    S_t = u+\sum_{j=1}^v\bigg\lceil\frac{t\lambda_j}{m}\bigg\rceil+\sum_{k=1}^w\bigg\lceil\frac{t\lambda_k^\prime}{m}\bigg\rceil-\bigg\lfloor\frac{t\lambda_0}{m}\bigg\rfloor,
\]
for $1\le t\le m-1$ and 
\[
    V_F = \{\lambda_j\mid f_j(t)\ge 0 \text{ for all }1\le t\le m-1\}. 
\]

We are now able to present the construction of non-special divisors derived from Kummer extensions.

\begin{theorem}\label{non-special_divisors_cons}
    Let $F/\mathbb{F}_q(x)$ be the Kummer extension defined by Equation (\ref{kum_diff}). For $1\le t\le m-1$, we define
    \[
        \ell_t=S_t,\;\;s_t=\begin{cases}
            \ell_t-\ell_{t+1}, &\text{if}\;\; t\neq m-1,\\
            \ell_{m-1}-1, &\text{if}\;\; t=m-1.
        \end{cases}
    \]
    Suppose that $0\le v,w\le \big\lfloor\frac{\lambda_0}{m}\big\rfloor$, $\lambda_j\in V_F$ for all $1\le j\le v$ and $s_t\ge0$ for $1\le t\le m-2$. Then the divisor 
    \[
        A=\sum_{t=1}^{m-2}t\sum_{l=1}^{s_t}Q_{t_l}+(m-1)\sum_{l^\prime=1}^{s_{m-1}}P_{l^\prime}
    \]
    is effective and non-special of degree $g$ with 
    \[
        P_{l^\prime}\in \{Q_i\}\cup\{Q_j^\prime\}\;\;\text{and}\;\;Q_{t_l}\in\{Q_i\}\backslash\{P_{l^\prime}\}.
    \]
\end{theorem}

Before we provide the proof of Theorem \ref{non-special_divisors_cons}, it is essential to introduce the following lemmas concerning the explicit description of $\Gamma^+(Q_i,Q_j^\prime)$ and certain properties associated with the floor and ceiling functions.

\begin{lemma}\cite[Prop. 3.12]{castellanosWeierstrassSemigroupsPure2024}\label{mini_gen_set_1}
    Let $F/\mathbb{F}_q(x)$ be a Kummer extension defined by Equation (\ref{kum_diff}). Then for $1\le i\le u$, we have 
    \[
        \Gamma^+(Q_i)=\mathbb{N}\,\backslash\Big\{m\mu+t \,\Big|\, 1\le t\le m-1,\,0\le \mu\le S_t-2\Big\}, 
    \]
    and for $1\le j\le v$, we have 
    \[
        \Gamma^+(Q_j^\prime)=\mathbb{N}\,\backslash\Big\{m\mu+t \,\Big|\, 1\le t\le m-1,\,0\le \mu\le f_j(t)\Big\}. 
    \]
\end{lemma}

\begin{lemma}\cite[Cor. 3.5]{castellanos2025generalizedweierstrasssemigroupsarbitrary}\label{mini_gen_set_mul}
    Let $F/\mathbb{F}_q(x)$ be a Kummer extension defined by Equation (\ref{kum_diff}). Let $\ell_1+\ell_2=\ell$ and $2\le\ell\le u+v$. For $1\le \nu^\prime\le \ell_2$, let $t_{j_{\nu^\prime}} = t\lambda_{j_{\nu^\prime}} \bmod m$. Then  
    \begin{align*}
        \Gamma^+(Q_{i_1},\cdots,Q_{i_{\ell_1}},Q_{j_1}^\prime,\cdots,&Q_{j_{\ell_2}}^\prime)=\Big\{(m\mu_1+t,\cdots,m\mu_{\ell_1}+t,m\mu^\prime_1+t_{j_1},\cdots,m\mu_{\ell_2}^\prime+t_{j_{\ell_2}})\,\Big|\\
        & 1\le t \le m-1,\,\mu_\nu\ge0,\,\mu_{\nu^\prime}^\prime\ge0,\,\sum_{\nu=1}^{\ell_1}\mu_\nu+\sum_{\nu^\prime=1}^{\ell_2}\mu_{\nu^\prime}^\prime=S_t -\ell\Big\}.
    \end{align*}
\end{lemma}

\begin{lemma}\cite[Lem. 4.1]{castellanosWeierstrassSemigroupsPure2024}\label{inter}
    Let $a$ and $b$ be elements in $\mathbb{R}$. The following statements holds:
    \begin{itemize}
        \item[(i)] $\lfloor-a\rfloor = -\lceil a\rceil$.
        \item[(ii)] $\lceil a\rceil - \lfloor a \rfloor = \begin{cases}
            0, &\text{if}\;\;a\in\mathbb{Z},\\
            1, &\text{if}\;\;a\notin\mathbb{Z}.
        \end{cases}$
        \item[(iii)] If $a$ and $b$ are positive integers, then
        \[
            \sum_{k=1}^{b-1}\bigg\lfloor\frac{ka}{b}\bigg\rfloor=\frac{(a-1)(b-1)+(a,b)-1}{2}.
        \]
    \end{itemize}
\end{lemma}

{
\raggedright
{\it Proof of Theorem \ref{non-special_divisors_cons}.} Note that
}
    \begin{align*}
        s_{m-1}& =\ell_{m-1}-1=S_{m-1}-1= u+\sum_{j=1}^v\bigg\lceil\frac{(m-1)\lambda_j}{m}\bigg\rceil+\sum_{k=1}^w\bigg\lceil\frac{(m-1)\lambda_k^\prime}{m}\bigg\rceil-\bigg\lfloor\frac{(m-1)\lambda_0}{m}\bigg\rfloor-1 \\
        & =u+\sum_{j=1}^v\bigg\lceil\lambda_j-\frac{\lambda_j}{m}\bigg\rceil+\sum_{k=1}^w\bigg\lceil\lambda_k^\prime-\frac{\lambda_k^\prime}{m}\bigg\rceil-\bigg\lfloor\lambda_0-\frac{\lambda_0}{m}\bigg\rfloor-1 \\
        & =u+\sum_{j=1}^v\lambda_j+\sum_{k=1}^w\lambda_k^\prime-\lambda_0-\bigg\lfloor-\frac{\lambda_0}{m}\bigg\rfloor-1 \\
        & = -\bigg\lceil-\frac{\lambda_0}{m}\bigg\rceil= \bigg\lfloor\frac{\lambda_0}{m}\bigg\rfloor\ge 0. 
    \end{align*}
    So we conclude that $s_t\ge0$ and the sequence $\ell_t$ is monotonically decreasing for $1\le t\le m-1$. On the other hand, 
    \[
        \#{\rm supp} (A)=\ell_1-1=S_1-1=u+v+w-\bigg\lfloor\frac{\lambda_0}{m}\bigg\rfloor-1< u+v.
    \]
    Therefore, the divisor $A$ is well-defined. Following this, we will prove that $\deg A=g$. Since 
    \[
        \# \{1\le t\le m-1\mid m \text{ divides } t\lambda \}=(\lambda,m)-1,\text{ for }1<\lambda<m,
    \]
    we have
    \begin{align*}
        \deg A&=\sum_{t=1}^{m-1}ts_t=\sum_{t=1}^{m-1}\ell_t-(m-1)=\sum_{t=1}^{m-1}S_t-(m-1) \\
        &=\sum_{t=1}^{m-1}\bigg(u+\sum_{j=1}^v\bigg\lceil\frac{t\lambda_j}{m}\bigg\rceil+\sum_{k=1}^w\bigg\lceil\frac{t\lambda_k^\prime}{m}\bigg\rceil-\bigg\lfloor\frac{t\lambda_0}{m}\bigg\rfloor\bigg)-(m-1) \\
        &=(m-1)u - (m-1) + \sum_{j=1}^v\sum_{t=1}^{m-1}\bigg\lceil\frac{t\lambda_j}{m}\bigg\rceil + \sum_{k=1}^w\sum_{t=1}^{m-1}\bigg\lceil\frac{t\lambda_k^\prime}{m}\bigg\rceil - \sum_{t=1}^{m-1}\bigg\lfloor\frac{t\lambda_0}{m}\bigg\rfloor \\
        &= (m-1)u - (m-1) -\sum_{t=1}^{m-1}\bigg\lfloor\frac{t\lambda_0}{m}\bigg\rfloor + \sum_{j=1}^v\bigg( m-1+\sum_{t=1}^{m-1}\bigg\lfloor\frac{t\lambda_j}{m}\bigg\rfloor \bigg) \\
        &\quad + \sum_{k=1}^w\bigg( m-(\lambda_k^\prime,m)+\sum_{t=1}^{m-1}\bigg\lfloor\frac{t\lambda_k^\prime}{m}\bigg\rfloor \bigg)\qquad\big(\text{By Lemma \ref{inter}(ii)}\big)\\
        &= (m-1)u - (m-1) - \frac{(m-1)(\lambda_0-1)}{2} + \sum_{j=1}^v\frac{(m-1)(\lambda_j+1)}{2}\\
        &\quad + \sum_{k=1}^w\frac{(m-1)(\lambda_k^\prime-1)-(\lambda_k^\prime,m)-1+2m}{2}\qquad\big(\text{By Lemma \ref{inter}(iii)}\big)\\
        &= \frac{(m-1)(u+v-1)+mw-\sum_{k=1}^w(\lambda_k^\prime,m)}{2}=g.
    \end{align*}
    Hence, $A$ is an effective divisor of degree $g$. We shall now prove that $A$ is non-special. 

    For any $1\le t\le m-1$, if $s_t>0$, then $\ell_1\ge \ell_2\ge\cdots\ge \ell_t=S_t\ge2$ since the sequence $\ell_t$ is monotonically decreasing. By Lemma \ref{mini_gen_set_1}, we have
    \[
        t<\xi\text{ for any }\xi\in\Gamma^+(Q_{i}).
    \]
    On the other hand, since $\lambda_j\in V_F$ for all $1\le j\le v$, we also have 
    \[
        m-1 < \xi\text{ for any }\xi\in\Gamma^+(Q_{j}^\prime).
    \]
    Suppose that ${\rm supp}(A)=\{Q_{i_\nu}\}_{\nu=1}^{T_1}\cup\{Q_{j_{\nu^\prime}}^\prime\}_{\nu^\prime=1}^{T_2} $ with $T_1+T_2=\ell_1-1$. Let $\theta = (\theta_i)_{i=1}^{\ell_1-1}\in\mathbb{N}^{\ell_1-1}$ such that 
    \begin{equation}\label{A_specific}
        A=\sum_{\nu=1}^{T_1}\theta_\nu Q_{i_\nu}+\sum_{\nu^\prime=1}^{T_2}\theta_{T_1+\nu^\prime}Q_{j_{\nu^\prime}}^\prime,
    \end{equation}
    where $\theta_\nu\le m-1$ for $1\le \nu\le T_1$ and $\theta_{T_1+\nu^\prime}=m-1$ for $1\le \nu^\prime\le T_2$ by the definition of $A$. Therefore, for any $\xi \in \Gamma^+(Q_{i_\nu})\big(\text{or }\Gamma^+(Q_{j_{\nu^\prime}}^\prime)\big)$, we have $\iota_{\{\nu\}}(\xi)\big(\text{or }\iota_{\{T_1+\nu^\prime\}}(\xi)\big)\nleq \theta$. 

    Subsequently, we consider $\xi=(\xi_i)\in\Gamma^+(P_l\mid l\in I\subseteq [\ell_1-1])$ where $\{P_1,\cdots,P_{T_1^\prime}\}\subseteq\{Q_{i_\nu}\}_{\nu=1}^{T_1^\prime}$, $\{P_{T_1^\prime+1},\cdots$, $P_{T_1^\prime+T_2^\prime}\}\subseteq\{Q_{j_{\nu^\prime}}^\prime\}_{\nu^\prime=1}^{T_2^\prime}$ and $T_1^\prime+T_2^\prime=|I|$. By Lemma \ref{mini_gen_set_mul}, we know that if $|I|\neq \ell_t$ for all $1\le t\le m-1$, then there exists $\xi_i>m$ for some $i$. Therefore, $\iota_I(\xi)\nleq \theta$. For this reason, we can assume that $|I|= \ell_t$ for some $t\in[m-1]$. Then we have
    \[
        \xi = (t,\cdots,t,t_{j_1},\cdots,t_{j_{T_2^\prime}}),
    \]
    where $t_{j_{\nu^\prime}}\le m-1$ for all $1\le \nu^\prime\le T_2^\prime$. Now we prove that $\iota_I(\xi)\nleq \theta$. Suppose, on the contrary, that there are $|I|-T_2^\prime=\ell_t-T_2^\prime$ places in $\{Q_{i_\nu}\}_{\nu=1}^{T_1}$ such that $\theta_\nu\ge t$ in Equation (\ref{A_specific}). By the definition of $A$, there are at most
    \[
        \sum_{l=t}^{m-2}s_l+s_{m-1}-T_2=\ell_t-1-T_2
    \]
    places in $\{Q_{i_\nu}\}_{\nu=1}^{T_1}$ such that $\theta_\nu\ge t$ in Equation (\ref{A_specific}). Consequently, we have a contradiction since $\ell_t - 1 - T_2 < \ell_t - T_2^\prime$. Thus, by Lemma \ref{non-special}, we conclude that $A$ is non-special. The proof is completed.  \qed

\begin{remark}\label{non-special_divisors_cons_rem}
    {\rm (1)} Within the demonstration of Theorem \ref{non-special_divisors_cons},  it becomes apparent that the divisor $A$ remains well-defined under the condition $v=\big\lfloor\frac{\lambda_0}{m}\big\rfloor+1$. In this case, the divisor $A$ can be represented as
    \[
        A=\sum_{t=1}^{m-2}t\sum_{l=1}^{s_t}Q_{t_l}+ (m-1)\sum_{l=1}^{s_{m-1}-k}Q_{(m-1)_l} + (m-1)\sum_{l^\prime=1}^{k}Q^\prime_{l^\prime},
    \]
    where $w\le k\le \big\lfloor\frac{\lambda_0}{m}\big\rfloor$. Similar to the proof in Theorem \ref{non-special_divisors_cons}, the above divisor $A$ is also effective and non-special of degree $g$

    {\rm (2)} Theorem \ref{non-special_divisors_cons} aligns with the result presented in \cite{Explicit_Non-special_Divisors_of_Small_Degree} for the case where $v=w=0$. While \cite{Explicit_Non-special_Divisors_of_Small_Degree} established the necessity as well, our results can be applied to Kummer extensions that are different from those in \cite{Explicit_Non-special_Divisors_of_Small_Degree}. 
\end{remark}

By combining Lemma \ref{non_special_g-1} and Theorem \ref{non-special_divisors_cons}, we can derive some non-special divisors of degree $g-1$ on some Kummer extensions in the following examples. 

\begin{example}\label{non-special_divisors_cons_ex}
    {\rm (1)} Consider the subcover of the BM curve in \cite[Rem. 3.2]{mendozaExplicitEquationsMaximal2022} defined by the affine equation
    \[
        W^m = U^d(U^d-1)\bigg( \frac{1-U^{d(q-1)}}{U^d-1} \bigg)^{q+1},\quad d\mid q+1\,\text{ and }\,m\mid q^2-q+1.
    \]Let $q = 5$, $m = 7$ and $d = 1$. Applying the morphism $\varphi(x,y) = (U,(U^5-U)/W)$ to the above equation, then we obtain a birationally equivalent curve
    \[
        y^7 = (x-2)(x-3)(x-4)\cdot x^6(x-1)^6.
    \]
    Subsequently, we have $\ell_i = 3$ for all $1\le i\le 6$ and $V_F = \{6\}$. Let $Q_i\,(2\le i\le 4)$ be the rational place corresponding to the zero of $x-i$, $Q_j^\prime\,(0\le j\le 1)$ the rational place corresponding to the zero of $x-j$ and $Q_\infty$ the unique place at infinity. By Theorem \ref{non-special_divisors_cons}, the divisor 
    \[
        A = 6(Q_2+Q_3) \,(\text{or } 6(Q_0^\prime+Q_1^\prime))
    \]
    is effective and non-special of degree $g=12$. By Lemma \ref{non_special_g-1}, the divisor $A-Q_\infty$ is non-special of degree~$11$.  

    {\rm (2)} Consider the Kummer extension $F/\mathbb{F}_q(x)$ defined by 
    \[
        y^6 = (x-\alpha_1)(x-\alpha_2)(x-\alpha_3)\cdot(x-\beta)^5\cdot(x-\gamma)^3, 
    \]
    where $\alpha_i,\,\beta,\,\gamma\in\mathbb{F}_q \,(1\le i\le 3)$ are pairwise distinct elements. Then we obtain the values of $\ell_i\,(1\le i\le 5)$ in Table~\ref{tab:ell_i_non-special_divisors_cons_ex_2} and $V_F = \{5\}$.
    \begin{table}[ht]
	\centering
	\begin{tabular}{ccccc}
		\hline
        $\ell_1$ & $\ell_2$ & $\ell_3$ & $\ell_4$ & $\ell_5$  \\
        \hline
		4      & 3      & 3      & 2      & 2     \\
        \hline
	\end{tabular}
	\caption{The values of $\ell_i$}
	\label{tab:ell_i_non-special_divisors_cons_ex_2}
    \end{table}
    Let $Q_i\,(1\le i\le 3)$ be the rational place corresponding to the zero of $x-\alpha_i$, $Q_\beta^\prime$ the rational place corresponding to the zero of $x-\beta$ and $Q_\infty$ the unique place at infinity. By Theorem \ref{non-special_divisors_cons}, the divisor
    \[
        A = Q_1 + 3Q_2 + 5 Q_3\,(\text{or } Q_1 + 3Q_2 + 5Q_\beta^\prime )
    \]
    is effective and non-special of degree $g=9$. By Lemma \ref{non_special_g-1}, the divisor $A-Q_\infty$ is non-special of degree $8$. 

    {\rm (3)} Consider the Kummer extension $F/\mathbb{F}_q(x)$ defined by 
    \[
        y^4 = x \cdot(x-\alpha_1)\cdot(x - \alpha_2)\cdot(x-1)^2,  
    \]
    where $\alpha_i\in\mathbb{F}_q\backslash\{0,1\} \,(1\le i\le 2)$ are pairwise distinct elements. Then we obtain the values of $\ell_i\,(1\le i\le 3)$ in Table \ref{tab:ell_i_non-special_divisors_cons_ex_3}.
    \begin{table}[ht]
	\centering
	\begin{tabular}{ccc}
		\hline
        $\ell_1$ & $\ell_2$ & $\ell_3$   \\
        \hline
		3      & 2      & 2       \\
        \hline
	\end{tabular}
	\caption{The values of $\ell_i$}
	\label{tab:ell_i_non-special_divisors_cons_ex_3}
    \end{table}
    Let $Q_i$ be the only rational place corresponding to the zero of $x-\alpha_i$, $Q_0$ the only rational place corresponding to the zero of $x$ and $Q_\infty$ the unique place at infinity. By Theorem~\ref{non-special_divisors_cons}, the divisor
    \[
        A = Q_0 + 3Q_1,
    \]
    is effective and non-special of degree $g=4$. By Lemma \ref{non_special_g-1}, the divisor $A-Q_2$ is non-special of degree~$3$. 
\end{example}

\section{Several General Constructions of LCPs of AG Codes via Kummer Extensions}\label{gen_con_LCP}

In this section, we provide several general constructions of LCPs of AG Codes via Kummer extensions. Specifically, one form of constructions involves Kummer extensions $F/\mathbb{F}_q(x)$ defined by Equation (\ref{kum_diff}) with $w=0$, while the other involves Kummer extensions $F/\mathbb{F}_q(x)$ with $v=0$.

\subsection{General Constructions via Kummer Extensions with $w=0$}

In the following, we introduce two general constructions of LCPs of AG Codes via Kummer extensions $F/\mathbb{F}_q(x)$ defined by Equation (\ref{kum_diff}) with $w=0$.

\begin{theorem}\label{con_w=0_1}
    Let $F/\mathbb{F}_q(x)$ be the Kummer extension defined by Equation (\ref{kum_diff}) with $w=0$. Suppose that $a_1,\cdots,a_\tau$ are elements of $\mathbb{F}_q$ such that $P_{a_i}$ is a rational place in $\mathbb{P}_{\mathbb{F}_q(x)}$ splitting completely in the extension $F/\mathbb{F}_q(x)$, denoted by $P_{a_ib_j}$ for $1\le j \le m$. Let
    \[
        D = \sum_{i=1}^\tau \sum_{j=1}^m P_{a_ib_j}
    \]
    and $n=\deg D=\tau m$. Keep the notations and assumptions in Theorem \ref{non-special_divisors_cons}. Let $s_0=0$ and $s_t^\prime=\sum_{i=0}^{t-1}s_i$ for $1\le t\le m-1$. For $m\le s <(n-g+2)/\lambda_0$, we define the divisors
    \begin{align*}
        & G = \sum_{t=1}^{m-1}t\sum_{l=1}^{s_t}Q_{l+s_t^\prime} + (n-\lambda_0s)Q_\infty,\;\text{ and }\\
        & H = \sum_{t=1}^{m-1}(s+m-1-t)\sum_{l=1}^{s_t}Q_{l+s_t^\prime} + (s+m-1)\sum_{l=1}^{s_{m-1}-v}Q_{u+v-\lfloor\frac{\lambda_0}{m}\rfloor-1+l}\\
        &\qquad + \sum_{l=1}^{v}(\lambda_l s+m-1)Q^\prime_l+ (s-1)Q_u - Q_\infty . 
    \end{align*}
    Then the pair $(\mathcal{C}_{\mathcal{L}}(D,G),\mathcal{C}_{\mathcal{L}}(D,H))$ is an LCP of AG Codes with parameters
    \begin{gather*}
            [n,n-\lambda_0s+1,\ge \lambda_0s-g]_q\text{ and }[n,\lambda_0s-1,\ge n-\lambda_0s-g+2]_q
    \end{gather*}
    respectively and the security parameter is $d(\mathcal{C}_{\mathcal{L}}(D,G))$. 
\end{theorem}

Before we provide the proof of Theorem \ref{con_w=0_1}, we need to introduce a characteristic of the ceiling function. 

\begin{lemma}\label{ceil_prop}
    Let $\lambda,m\in\mathbb{N}$ and $1\le \lambda<m$. For $1\le k\le m-2$, 
    \[
        \bigg\lceil\frac{(k+1)\lambda}{m}\bigg\rceil-\bigg\lceil\frac{k\lambda}{m}\bigg\rceil=\begin{cases}
            1, &\text{when}\;\;k=\big\lfloor\frac{l m}{\lambda}\big\rfloor,\;1\le l\le\lambda-1,\\
            0, &\text{otherwise}.
        \end{cases}
    \]
\end{lemma}
\begin{proof}
    For some $1\le l\le\lambda-1$, if $k=\big\lfloor\frac{l m}{\lambda}\big\rfloor$, then we have
    \[
        \frac{(l-1)m}{\lambda}<k\le\frac{l m}{\lambda}<k+1<\frac{(l+1)m}{\lambda}\Longrightarrow l-1<\frac{k\lambda}{m}\le l<\frac{(k+1)\lambda}{m}<l+1,
    \]
    which implies that $\big\lceil\frac{(k+1)\lambda}{m}\big\rceil-\big\lceil\frac{k\lambda}{m}\big\rceil=(l+1)-l=1$. On the other hand, if 
    \[
        \bigg\lfloor\frac{l m}{\lambda}\bigg\rfloor<k<k+1\le\bigg\lfloor\frac{(l+1)m}{\lambda}\bigg\rfloor,
    \]
    then we have
    \[
        \frac{l m}{\lambda}<\bigg\lfloor\frac{l m}{\lambda}\bigg\rfloor+1\le k<k+1\le\frac{(l+1)m}{\lambda}\Longrightarrow l<\frac{k\lambda}{m}<\frac{(k+1)\lambda}{m}\le l+1.
    \]
    Thus $\big\lceil\frac{(k+1)\lambda}{m}\big\rceil-\big\lceil\frac{k\lambda}{m}\big\rceil=0$. At last, if $\big\lfloor\frac{m}{\lambda}\big\rfloor\ge2$, then 
    \[
        1\le k<k+1\le\bigg\lfloor\frac{m}{\lambda}\bigg\rfloor\Longrightarrow \frac{\lambda}{m}\le \frac{k\lambda}{m}<\frac{(k+1)\lambda}{m}\le 1. 
    \]
    It follows that $\big\lceil\frac{(k+1)\lambda}{m}\big\rceil-\big\lceil\frac{k\lambda}{m}\big\rceil=0$. The proof is completed. 
\end{proof}

{
\raggedright
{\it Proof of Theorem \ref{con_w=0_1}.} We aim to prove our result by verifying the conditions presented in Lemma~\ref{LCP_criterion}. Firstly, we can easily obtain the following assertions:
}
    \begin{itemize}
        \item ${\rm supp}(D)\cap{\rm supp}(G)={\rm supp}(D)\cap{\rm supp}(H)=\varnothing$.
        \item By $m\le s <(n-g+2)/(u+\sum_{l=1}^v\lambda_l)$, we have $2g-2<\deg G,\deg H< n$. 
    \end{itemize}
    It follows that the divisors $D,G,H$ can form two AG Codes $\mathcal{C}_{\mathcal{L}}(D,G)$ and $\mathcal{C}_{\mathcal{L}}(D,H)$. Subsequently, we can derive that
    \begin{align*}
            \gcd(G,H)&=\sum_{t=1}^{m-1}t\sum_{l=1}^{s_t}Q_{l+s_t^\prime}-Q_\infty, \text{ and }\\
            {\rm lmd}(G,H)&=\sum_{t=1}^{m-1}(s+m-1-t)\sum_{l=1}^{s_t}Q_{l+s_t^\prime} + (s+m-1)\sum_{l=1}^{s_{m-1}-v}Q_{u+v-\lfloor\frac{\lambda_0}{m}\rfloor-1+l}\\ 
            &\quad +\sum_{l=1}^{v}(\lambda_l s+m-1)Q^\prime_l + (s-1)Q_u+(n-\lambda_0s)Q_\infty.
    \end{align*}
    Here we claim that $s_t=s_{m-1-t}$ for $1\le t\le m-2$. Indeed, we have
    \begin{equation}\label{s_t_w=0}
        s_t = \sum_{j=1}^v\bigg(\bigg\lceil\frac{t\lambda_j}{m}\bigg\rceil-\bigg\lceil\frac{(t+1)\lambda_j}{m}\bigg\rceil\bigg)+\bigg(\bigg\lfloor\frac{(t+1)\lambda_0}{m}\bigg\rfloor-\bigg\lfloor\frac{t\lambda_0}{m}\bigg\rfloor\bigg).
    \end{equation}
    Since $m-1=\big\lfloor\frac{l m}{\lambda}\big\rfloor+\big\lfloor\frac{(\lambda-l) m}{\lambda}\big\rfloor$ for all $(\lambda,m)=1\,(\lambda\neq 1)$, $1\le l\le \lambda-1$ and Lemma \ref{ceil_prop}, we have
    \begin{gather*}
        \bigg\lceil\frac{t\lambda_j}{m}\bigg\rceil-\bigg\lceil\frac{(t+1)\lambda_j}{m}\bigg\rceil = \bigg\lceil\frac{(m-1-t)\lambda_j}{m}\bigg\rceil-\bigg\lceil\frac{(m-1-t+1)\lambda_j}{m}\bigg\rceil,\\
        \text{ and }\bigg\lfloor\frac{(t+1)\lambda_0}{m}\bigg\rfloor-\bigg\lfloor\frac{t\lambda_0}{m}\bigg\rfloor = \bigg\lfloor\frac{(m-1-t+1)\lambda_0}{m}\bigg\rfloor-\bigg\lfloor\frac{(m-1-t)\lambda_0}{m}\bigg\rfloor.
    \end{gather*}
    Therefore, the claim is correct. It follows that 
    \begin{align*}
        &\sum_{t=1}^{m-1}(m-1-t)\sum_{l=1}^{s_t}Q_{l+s_t^\prime} + (m-1)\sum_{l=1}^{s_{m-1}-v}Q_{u+v-\lfloor\frac{\lambda_0}{m}\rfloor-1+l} + \sum_{l=1}^{v}(m-1)Q^\prime_l \\
        &\quad = \sum_{t=1}^{m-2}(m-1-t)\sum_{l=1}^{s_{m-1-t}}Q_{l+s_t^\prime} + (m-1)\sum_{l=1}^{s_{m-1}-v}Q_{u+v-\lfloor\frac{\lambda_0}{m}\rfloor-1+l} + \sum_{l=1}^{v}(m-1)Q^\prime_l, 
    \end{align*}
    is a non-special divisor of degree $g$ by Theorem \ref{non-special_divisors_cons}. Then combining Riemann-Roch Theorem and Theorem~\ref{non-special_divisors_cons} yields
        \begin{align*}
            \ell(G)+\ell(H)&=\deg G + \deg H +2-2g\\
            &=g+(n-\lambda_0s) + g + us + \sum_{l=1}^v\lambda_l s - 2 +2 -2g \\
            &= n .
        \end{align*}
    
    In the following, we will prove that $\gcd(G,H)$ and ${\rm lmd}(G,H)-D$ are non-special. By Theorem \ref{non-special_divisors_cons} and Lemma~\ref{non_special_g-1}, we know that $\gcd(G,H)$ is a non-special divisor of degree $g-1$. Since 
    \[
        (y^s) = s\bigg(\sum_{i=1}^uQ_i+\sum_{j=1}^v\lambda_jQ_j^\prime\bigg)-s\lambda_0Q_\infty\,\text{ and }\,(h)=D-nQ_\infty,
    \]
    where $h=\prod_{i=1}^\tau(x-a_i)$, we derive that
    \begin{align*}
        {\rm lmd}(G,H)-D & = {\rm lmd}(G,H)-(h)-nQ_\infty\\
        &\sim \sum_{t=1}^{m-1}(s+m-1-t)\sum_{l=1}^{s_t}Q_{l+s_t^\prime} + (s+m-1)\sum_{l=1}^{s_{m-1}-v}Q_{u+v-\lfloor\frac{\lambda_0}{m}\rfloor-1+l}\\ 
        &\quad +\sum_{l=1}^{v}(\lambda_l s+m-1)Q^\prime_l + (s-1)Q_u-\lambda_0sQ_\infty\\
        & = \sum_{t=1}^{m-2}(m-1-t)\sum_{l=1}^{s_{m-1-t}}Q_{l+s_t^\prime} + (m-1)\sum_{l=1}^{s_{m-1}-v}Q_{u+v-\lfloor\frac{\lambda_0}{m}\rfloor-1+l}\\ 
        &\quad +\sum_{l=1}^{v}(m-1)Q^\prime_l -Q_u + (y^s)\\
        & \sim \sum_{t=1}^{m-2}(m-1-t)\sum_{l=1}^{s_{m-1-t}}Q_{l+s_t^\prime} + (m-1)\sum_{l=1}^{s_{m-1}-v}Q_{u+v-\lfloor\frac{\lambda_0}{m}\rfloor-1+l}\\ 
        &\quad +\sum_{l=1}^{v}(m-1)Q^\prime_l -Q_u. 
    \end{align*}
    By deduction above and Lemma \ref{non_special_g-1}, we know that ${\rm lmd}(G,H)-D$ is also non-special. Hence, by Lemma~\ref{LCP_criterion}, we conclude that $(\mathcal{C}_{\mathcal{L}}(D,G),\mathcal{C}_{\mathcal{L}}(D,H))$ is an LCP of AG Codes. 

    Now we prove that $\mathcal{C}_{\mathcal{L}}(D,H)^{\bot}\sim \mathcal{C}_{\mathcal{L}}(D,G)$. By Lemma~\ref{dual_AG Code}, we have
    \[
        \mathcal{C}_{\mathcal{L}}(D,G)^{\bot}=\mathcal{C}_{\mathcal{L}}(D,D-G+(\eta))=\mathcal{C}_{\mathcal{L}}(D,D-G+(dh)-(h)),
    \]
    where $\eta = dh/h$ and $h=\prod_{i=1}^\tau(x-a_i)$. From \cite[Chap. 3 and 4]{Algebraic_Function_Fields_and_Codes}, the different of $F/\mathbb{F}_q(x)$ is given by
    \[
        {\rm Diff}(F/\mathbb{F}_q(x)) = \sum_{i=1}^u(m-1)Q_i+\sum_{j=1}^v(m-1)Q_j^\prime+(m-1)Q_\infty,
    \]
    and
    \begin{align*}
        (\eta) &= (dh) - (h) =  - (h) + \bigg(\frac{dh}{dx}\bigg) + (dx)=  - D + n Q_\infty  + \bigg(\frac{dh}{dx}\bigg) -2(x)_\infty + {\rm Diff}(F/\mathbb{F}_q(x))\\
        & =  - D + n Q_\infty  + \bigg(\frac{dh}{dx}\bigg) - 2m Q_\infty +\sum_{i=1}^u(m-1)Q_i+\sum_{j=1}^v(m-1)Q_j^\prime+(m-1)Q_\infty\\
        & = - D + \bigg(\frac{dh}{dx}\bigg) + \sum_{i=1}^u(m-1)Q_i+\sum_{j=1}^v(m-1)Q_j^\prime+(n-m-1)Q_\infty.
    \end{align*}
    It follows that
    \begin{align*}
        G+H & = \sum_{i=1}^{u-1}(m-1)Q_i + \sum_{j=1}^v(m-1)Q_j^\prime - Q_u + (n-1)Q_\infty + (y^s)\\
        & =  \sum_{i=1}^{u}(m-1)Q_i + \sum_{j=1}^v(m-1)Q_j^\prime + (n-m-1)Q_\infty + \bigg(\frac{y^s}{x-\alpha_u}\bigg) \\
        & \sim D + (\eta). 
    \end{align*}
    Hence, 
    \[
        \mathcal{C}_{\mathcal{L}}(D,G)^{\bot}=\mathcal{C}_{\mathcal{L}}(D,D-G+(\eta))\sim \mathcal{C}_{\mathcal{L}}(D,H)\;\text{ i.e. }\;\mathcal{C}_{\mathcal{L}}(D,H)^{\bot}\sim \mathcal{C}_{\mathcal{L}}(D,G).
    \]
    Therefore, the security parameter of $(\mathcal{C}_{\mathcal{L}}(D,G),\mathcal{C}_{\mathcal{L}}(D,H))$ is $d(\mathcal{C}_{\mathcal{L}}(D,G))$. The proof is completed.   \qed

\begin{theorem}\label{con_w=0_2}
    Let $F/\mathbb{F}_q(x)$ be the Kummer extension defined by Equation (\ref{kum_diff}) with $w=0$. Suppose that $a_1,\cdots,a_\tau$ are elements of $\mathbb{F}_q$ such that $P_{a_i}$ is a rational place in $\mathbb{P}_{\mathbb{F}_q(x)}$ splitting completely in the extension $F/\mathbb{F}_q(x)$, denoted by $P_{a_ib_j}$ for $1\le j \le m$. Let
    \[
        D = \sum_{i=1}^\tau \sum_{j=1}^m P_{a_ib_j}
    \]
    and $n=\deg D=\tau m$. Keep the notations and assumptions in Theorem \ref{non-special_divisors_cons}. Let $s_0=0$ and $s_t^\prime=\sum_{i=0}^{t-1}s_i$ for $1\le t\le m-1$. For $m\le s <(n-g+1)/(\lambda_0-1)$, we define the divisors
    \begin{align*}
        & G = \sum_{t=1}^{m-1}t\sum_{l=1}^{s_t}Q_{l+s_t^\prime} + sQ_u +  (n-\lambda_0s-1)Q_\infty,\text{ and }\\
        & H = \sum_{t=1}^{m-1}(s+m-1-t)\sum_{l=1}^{s_t}Q_{l+s_t^\prime} + (s+m-1)\sum_{l=1}^{s_{m-1}-v}Q_{u+v-\lfloor\frac{\lambda_0}{m}\rfloor-1+l}\\
        &\qquad + \sum_{l=1}^{v}(\lambda_l s+m-1)Q^\prime_l -Q_u. 
    \end{align*}
    Then the pair $(\mathcal{C}_{\mathcal{L}}(D,G),\mathcal{C}_{\mathcal{L}}(D,H))$ is an LCP of AG Codes with parameters
    \begin{gather*}
            [n,n-(\lambda_0-1)s,\ge (\lambda_0-1)s-g+1]_q\text{ and }[n,(\lambda_0-1)s,\ge n-(\lambda_0-1)s-g+1]_q
    \end{gather*}
    respectively and the security parameter is $d(\mathcal{C}_{\mathcal{L}}(D,G))$. 
\end{theorem}
\begin{proof}
    It should be noted that we merely alter the positions of certain rational places, and we can derive that 
    \[
        \gcd (G,H) = \sum_{t=1}^{m-1}t\sum_{l=1}^{s_t}Q_{l+s_t^\prime} - Q_u,
    \]
    and 
    \[
        {\rm lmd}(G,H)-D \sim \sum_{t=1}^{m-2}(m-1-t)\sum_{l=1}^{s_{m-1-t}}Q_{l+s_t^\prime} + (m-1)\bigg(\sum_{l=1}^{s_{m-1}-v}Q_{u+v-\lfloor\frac{\lambda_0}{m}\rfloor-1+l}+\sum_{l=1}^{v}Q^\prime_l\bigg) -Q_\infty.
    \]
    The rest of the proof is completely similar to the proof of Theorem \ref{con_w=0_1}. 
\end{proof}

\begin{remark}\label{con_w=0_rem}
    {\rm (1)} By Remark \ref{non-special_divisors_cons_rem}(1), we can also construct LCPs of AG Codes when $v=\big\lfloor\frac{\lambda_0}{m}\big\rfloor+1$ in Theorems~\ref{con_w=0_1} and \ref{con_w=0_2}. In this case, 
    \begin{itemize}
        \item if we define the divisors
        \begin{align*}
        & G = \sum_{t=1}^{m-1}t\sum_{l=1}^{s_t}Q_{l+s_t^\prime} + (n-\lambda_0s)Q_\infty,\;\text{ and }\\
        & H = \sum_{t=1}^{m-1}(s+m-1-t)\sum_{l=1}^{s_t}Q_{l+s_t^\prime} + \sum_{l=1}^{s_{m-1}}(\lambda_{l} s+m-1)Q^\prime_{l} + (\lambda_{\lfloor\frac{\lambda_0}{m}\rfloor+1}s-1)Q^\prime_{\lfloor\frac{\lambda_0}{m}\rfloor+1} - Q_\infty. 
    \end{align*}
    in Theorem \ref{con_w=0_1},
    \item and we define the divisors
    \begin{align*}
        & G = \sum_{t=1}^{m-1}t\sum_{l=1}^{s_t}Q_{l+s_t^\prime} + \lambda_{\lfloor\frac{\lambda_0}{m}\rfloor+1}sQ^\prime_{\lfloor\frac{\lambda_0}{m}\rfloor+1} +  (n-\lambda_0s-1)Q_\infty,\text{ and }\\
        & H = \sum_{t=1}^{m-1}(s+m-1-t)\sum_{l=1}^{s_t}Q_{l+s_t^\prime}+ \sum_{l=1}^{s_{m-1}}(\lambda_{l} s+m-1)Q^\prime_{l} -Q^\prime_{\lfloor\frac{\lambda_0}{m}\rfloor+1}. 
    \end{align*}
    in Theorem \ref{con_w=0_2},
    \end{itemize}
    then we can also derive certain specific LCPs of AG Codes. 

    {\rm (2)} Theorems \ref{con_w=0_1} and \ref{con_w=0_2} align with the results presented in \cite{Linear_Complementary_Dual_codes_and_Linear_Complementary_Pairs_of_AG_codes_in_function_fields} for the case where $v=0$ and $m>\lambda_0$. It should be noted that our results can be applied to the Kummer extensions with $m<\lambda_0$. 
\end{remark}

\subsection{General Construction via Kummer Extensions with $v=0$}

In the following, we provide a general construction of LCPs of AG Codes via Kummer extensions $F/\mathbb{F}_q(x)$ defined by Equation (\ref{kum_diff}) with $v=0$.

\begin{theorem}\label{con_v=0}
    Let $F/\mathbb{F}_q(x)$ be the Kummer extension defined by Equation (\ref{kum_diff}) with $v=0$ and $\lambda_k^\prime\mid m$ for all $1\le k\le w$. Suppose that $a_1,\cdots,a_\tau$ are elements of $\mathbb{F}_q$ such that $P_{a_i}$ is a rational place in $\mathbb{P}_{\mathbb{F}_q(x)}$ splitting completely in the extension $F/\mathbb{F}_q(x)$, denoted by $P_{a_ib_j}$ for $1\le j \le m$. Let
    \[
        D = \sum_{i=1}^\tau \sum_{j=1}^m P_{a_ib_j}
    \]
    and $n=\deg D=\tau m$. Keep the notations and assumptions in Theorem \ref{non-special_divisors_cons}. Let $s_0=0$, $s_t^\prime=\sum_{i=0}^{t-1}s_i$ for $1\le t\le m-1$ and 
    \[
        R_k = \sum_{Q\in\mathbb{P}_F,Q\mid P_{\gamma_k}} Q\; \text{ for all }1\le k\le w,
    \]
    where $P_{\gamma_k}$ is the rational place corresponding to $x-\gamma_k$ in $\mathbb{P}_{\mathbb{F}_q(x)}$. For $m\le s <(n-g+1)/\lambda_0$ and $\lambda_0 \equiv 1 \bmod m$, we define the divisors
    \begin{align*}
        G & = \sum_{t=1}^{m-1}t\sum_{l=1}^{s_t}Q_{l+s_t^\prime} + (n-\lambda_0s-1)Q_\infty\\
        H & = \sum_{t=1}^{m-1}(m-t)\sum_{l=1}^{s_t}Q_{l+s_t^\prime} +  s\bigg( \sum_{i=1}^uQ_i + \sum_{k=1}^wR_k \bigg)-Q_\infty. 
    \end{align*}
    Then the pair $(\mathcal{C}_{\mathcal{L}}(D,G),\mathcal{C}_{\mathcal{L}}(D,H))$ is an LCP of AG Codes with parameters
    \begin{gather*}
            [n,n-\lambda_0s,\ge \lambda_0s-g+1]_q\text{ and }[n,\lambda_0s,\ge n-\lambda_0s-g+1]_q
    \end{gather*}
    respectively and the security parameter is $d(\mathcal{C}_{\mathcal{L}}(D,G))$. 
\end{theorem}
\begin{proof}
    We aim to prove our result by verifying the conditions presented in Lemma~\ref{LCP_criterion}. Firstly, we can easily obtain the following assertions:
    \begin{itemize}
        \item By \cite[Prop. 3.7.3]{Algebraic_Function_Fields_and_Codes}, we know that $\deg R_k = \lambda_k^\prime$ for all $1\le k\le w$. Therefore, by $m\le s <(n-g+1)/\lambda_0$, we have $2g-2<\deg G,\deg H<n$.
        \item ${\rm supp}(D)\cap{\rm supp}(G)={\rm supp}(D)\cap{\rm supp}(H)=\varnothing$. 
    \end{itemize}
    It follows that the divisors $D,G,H$ can form two AG Codes $\mathcal{C}_{\mathcal{L}}(D,G)$ and $\mathcal{C}_{\mathcal{L}}(D,H)$. Subsequently, we can derive that
    \begin{align*}
            \gcd(G,H)&=\sum_{t=1}^{m-1}t\sum_{l=1}^{s_t}Q_{l+s_t^\prime}-Q_\infty, \text{ and }\\
            {\rm lmd}(G,H)&=\sum_{t=1}^{m-1}(m-t)\sum_{l=1}^{s_t}Q_{l+s_t^\prime} +  s\bigg( \sum_{i=1}^uQ_i + \sum_{k=1}^wR_k \bigg) + (n-\lambda_0s-1)Q_\infty.
    \end{align*}
    Here we claim that $s_t=s_{m-t}$ for $1\le t\le m-1$. Indeed, we have
    \begin{equation}\label{s_t_v=0}
        s_t = \sum_{k=1}^w\bigg(\bigg\lceil\frac{t\lambda_k^\prime}{m}\bigg\rceil-\bigg\lceil\frac{(t+1)\lambda_k^\prime}{m}\bigg\rceil\bigg)+\bigg(\bigg\lfloor\frac{(t+1)\lambda_0}{m}\bigg\rfloor-\bigg\lfloor\frac{t\lambda_0}{m}\bigg\rfloor\bigg).
    \end{equation}
    Since $m = \frac{l m}{\lambda}+\frac{(\lambda-l) m}{\lambda}$ for all $\lambda\mid m$ and $1\le l\le \lambda-1$ and Lemma \ref{ceil_prop}, we have
    \[
        \bigg\lceil\frac{t\lambda_k^\prime}{m}\bigg\rceil-\bigg\lceil\frac{(t+1)\lambda_k^\prime}{m}\bigg\rceil = \bigg\lceil\frac{(m-t)\lambda_k^\prime}{m}\bigg\rceil-\bigg\lceil\frac{(m-t+1)\lambda_k^\prime}{m}\bigg\rceil,
    \]
    and
    \[
        \bigg\lfloor\frac{(t+1)\lambda_0}{m}\bigg\rfloor-\bigg\lfloor\frac{t\lambda_0}{m}\bigg\rfloor = \bigg\lfloor\frac{\lambda_0}{m}\bigg\rfloor ,\;\text{ for all }1\le t\le m-1. 
    \]
    Therefore, the claim is correct. It follows that
    \[
        \sum_{t=1}^{m-1}(m-t)\sum_{l=1}^{s_t}Q_{l+s_t^\prime} = \sum_{t=1}^{m-1}(m-t)\sum_{l=1}^{s_{m-t}}Q_{l+s_t^\prime}, 
    \]
    is a non-special divisor of degree $g$ by Theorem \ref{non-special_divisors_cons}. Then by Riemann-Roch Theorem and Theorem \ref{non-special_divisors_cons}, we obtain
    \begin{align*}
        \ell(G) + \ell (H) & = \deg G + \deg H +2- 2g \\
        & = g+ (n-\lambda_0s-1) + g + s\bigg(u + \sum_{k=1}^w\lambda_k^\prime\bigg) -1 +2 -2g\\
        & = n . 
    \end{align*}
    Additionally, since $(y^s) = s\big(\sum_{i=1}^uQ_i+\sum_{k=1}^wR_k\big)-s\lambda_0Q_\infty$ and $(h)=D-nQ_\infty$, we have 
    \begin{align*}
        {\rm lmd}(G,H)- D &= {\rm lmd}(G,H)-(h)-nQ_\infty\\
        &\sim \sum_{t=1}^{m-1}(m-t)\sum_{l=1}^{s_t}Q_{l+s_t^\prime} +  s\bigg( \sum_{i=1}^uQ_i + \sum_{k=1}^wR_k \bigg) - (\lambda_0s+1)Q_\infty \\
        &= \sum_{t=1}^{m-1}(m-t)\sum_{l=1}^{s_{m-t}}Q_{l+s_t^\prime} - Q_\infty + (y^s) \\
        & \sim \sum_{t=1}^{m-1}(m-t)\sum_{l=1}^{s_{m-t}}Q_{l+s_t^\prime} - Q_\infty. 
    \end{align*}
    By Theorem \ref{non-special_divisors_cons} and Lemma \ref{non_special_g-1}, we know that ${\rm lmd}(G,H)-D$ and $\gcd (G,H)$ are non-special divisors of degree $g-1$. Hence, we conclude that $(\mathcal{C}_{\mathcal{L}}(D,G),\mathcal{C}_{\mathcal{L}}(D,H))$ is an LCP of AG Codes by Lemma \ref{LCP_criterion}. 

    Moreover, from \cite[Chap. 3 and 4]{Algebraic_Function_Fields_and_Codes}, the different of $F/\mathbb{F}_q(x)$ is given by
    \[
        {\rm Diff}(F/\mathbb{F}_q(x)) = \sum_{i=1}^u(m-1)Q_i+ \sum_{k=1}^w\bigg(\frac{m}{\lambda_k^\prime}-1\bigg)R_k +(m-1)Q_\infty. 
    \]
    Let $(\eta) = dh/h$ with $h=\prod_{i=1}^\tau(x-a_i)$. Then
    \[
        (\eta) =  - D + \bigg(\frac{dh}{dx}\bigg) + \sum_{i=1}^u(m-1)Q_i+ \sum_{k=1}^w\bigg(\frac{m}{\lambda_k^\prime}-1\bigg)R_k +(n-m-1)Q_\infty.
    \]
    Since for all $1\le k \le w$, 
    \[
        (x-\gamma_k) = \sum_{Q\in\mathbb{P}_F,Q\mid P_{\gamma_k}} \frac{m}{\lambda_k^\prime}Q - mQ_\infty= \frac{m}{\lambda_k^\prime}R_k - mQ_\infty,
    \]
    we have
    \begin{align*}
        G+H & = \sum_{t=1}^{m-1}m\sum_{l=1}^{s_t}Q_{l+s_t^\prime} + s\bigg( \sum_{i=1}^uQ_i + \sum_{k=1}^wR_k \bigg) + (n-\lambda_0s-2)Q_\infty\\
        & = \sum_{i=1}^{u+w-\lfloor\frac{\lambda_0}{m}\rfloor-1}(m-1)Q_i + \lambda_0 Q_\infty - \bigg( \sum_{i=u+w-\lfloor\frac{\lambda_0}{m}\rfloor}^u Q_i + \sum_{k=1}^w R_k \bigg) + (n-2)Q_\infty + (y^{s+1}) \\
        & = \sum_{i=1}^u(m-1)Q_i + \sum_{k=1}^w\bigg(\frac{m}{\lambda_k^\prime}-1\bigg)R_k + \bigg[ n+\lambda_0 -2 - \bigg(\bigg\lfloor\frac{\lambda_0}{m}\bigg\rfloor-w+1\bigg)m - wm \bigg]Q_\infty \\
        &\quad + \bigg(\frac{y^{s+1}}{\prod_{i=u+w-\lfloor\frac{\lambda_0}{m}\rfloor}^u (x-\alpha_i)\prod_{k=1}^w(x- \gamma_k)}\bigg) \\
        & \sim \sum_{i=1}^u(m-1)Q_i + \sum_{k=1}^w\bigg(\frac{m}{\lambda_k^\prime}-1\bigg)R_k + (n-m-1)Q_\infty\\
        & \sim D + (\eta). 
    \end{align*}
    Hence, by Lemma \ref{dual_AG Code}, 
    \[
        \mathcal{C}_{\mathcal{L}}(D,G)^{\bot}=\mathcal{C}_{\mathcal{L}}(D,D-G+(\eta))\sim \mathcal{C}_{\mathcal{L}}(D,H)\;\text{ i.e. }\;\mathcal{C}_{\mathcal{L}}(D,H)^{\bot}\sim \mathcal{C}_{\mathcal{L}}(D,G).
    \]
    Therefore, the security parameter of $(\mathcal{C}_{\mathcal{L}}(D,G),\mathcal{C}_{\mathcal{L}}(D,H))$ is $d(\mathcal{C}_{\mathcal{L}}(D,G))$. The proof is completed. 
\end{proof}

\begin{remark}
    It is important to highlight that when $w=v=0$, the results in Theorems \ref{con_w=0_1}, \ref{con_w=0_2}, and \ref{con_v=0} remain valid. Consequently, our results are applicable to a wide range of traditional curves.
\end{remark}

\section{Some Explicit Constructions of LCPs on Function Fields}\label{LCPs_of_AG Codes}

Let $r\ge3$ be an odd integer and $m\ge2$. We define the following affine equations
\[
    \mathcal{X}_m : W^m = U^d(U^d-1)\bigg( \frac{1-U^{d(q-1)}}{U^d-1} \bigg)^{q+1},\quad d\mid q+1,\quad m\mid q^2-q+1\,\text{ and }\,(m,q+1)=1,
\]
and 
\[
    \mathcal{Y}_m : W^m = \frac{1-U^{d(q+1)}}{U^d},\quad d\mid q-1,\quad m\,\bigg|\,\frac{q^r+1}{q+1}. 
\]
Let $F(\mathcal{X}_m)$ and $F(\mathcal{Y}_m)$ be the function fields of the curves $\mathcal{X}_m$ and $\mathcal{Y}_m$ respectively. From \cite{mendozaExplicitEquationsMaximal2022}, we know that these two curves are subcovers of the BM curve. Moreover, $F(\mathcal{X}_m)$ is maximal over $\mathbb{F}_{q^6}$ and $F(\mathcal{Y}_m)$ is maximal over $\mathbb{F}_{q^{2r}}$. In this section, we will provide several families of LCPs of AG Codes with exact parameters over these two subcovers of the BM curve at first. Subsequently, we will provide several families of LCPs of AG Codes over other function fields, such as elliptic function fields, hyperelliptic function fields and so on.

\subsection{LCPs of AG Codes over the Function Field $F(\mathcal{X}_m)$} 

Since $(m,q+1)=1$, there exist $\lambda,\rho\in\mathbb{Z}$ such that $1< \lambda<m$ and $\lambda\cdot (q+1)+\rho \cdot m=1$. Let $d\lambda\equiv d^\prime \bmod m$ with $1\le d^\prime<m$. Applying the morphism $\varphi(x,y) = \big(U,W^\lambda\cdot( \frac{1-U^{d(q-1)}}{U^d-1} )^\rho\cdot U^{-\frac{d\lambda-d^\prime}{m}}\big)$ to the curve $\mathcal{X}_m$, then we obtain a birationally equivalent curve
\[
    \mathfrak{C}:y^m = x^{d^\prime}(x^d-1)^\lambda\bigg( \frac{1-x^{d(q-1)}}{x^d-1} \bigg) = - \prod_{i=1}^{d(q-2)}(x-\alpha_i)\cdot\prod_{j=1}^d(x-\beta_j)^\lambda\cdot x^{d^\prime}, 
\]
where $\beta_j\in \{ a\in\mathbb{F}_{q^2}\mid a^d = 1 \}$ and $\alpha_i\in\{ a\in\mathbb{F}_{q^2}\mid a^{d(q-1)} = 1 \}\backslash\{\beta_j\}$. According to \cite[Cor. 4.5]{Algebraic_Geometry}, the function field of the curve $\mathfrak{C}$ is isomorphic to $F(\mathcal{X}_m)$. Therefore, avoiding abuse of notations, we could rewrite the affine curve of $F(\mathcal{X}_m)$ as follows:
\[
    \mathcal{X}_m : y^m = - \prod_{i=1}^{d(q-2)}(x-\alpha_i)\cdot\prod_{j=1}^d(x-\beta_j)^\lambda\cdot x^{d^\prime}. 
\]
The genus of $F(\mathcal{X}_m)$ is $g=\frac{d(m-1)(q-1)}{2}$. There are totally $q^6+2gq^3+1$ rational places in $F(\mathcal{X}_m)/\mathbb{F}_{q^6}$, since $F(\mathcal{X}_m)$ is maximal over $\mathbb{F}_{q^6}$. 

Note that, there exists $k\in\mathbb{Z}$ such that
\[
    (q-2)(q+1) = q^2-q-2 = q^2-q+1 -3 = km -3. 
\]
It follows that
\[
    (q-2+2\lambda)(-q-1) + (k-2\rho)m = [(q-2)(-q-1)+km] + [2\lambda(-q-1)-2\rho m] = 3-2 =1,
\]
i.e. $(q-2+2\lambda,m)=1$. So we can derive that 
\[
    (d(q-2)+d\lambda+d^\prime,m)=(d(q-2)+d\lambda+d\lambda,m) = (q-2+2\lambda,m)=1.
\]
Hence, there exists a unique place at infinity. Let $Q_i$ be the only rational place corresponding to the zero of $x-\alpha_i$, $Q_j^\prime$ the only rational place corresponding to the zero of $x-\beta_j$ and $Q_\infty$ the unique place at infinity. Moreover, the places which are ramified in the Kummer extension $F(\mathcal{X}_m)/\mathbb{F}_{q^6}(x)$ are exactly $Q_i\cap\mathbb{F}_{q^6}(x)$, $Q_j^\prime\cap\mathbb{F}_{q^6}(x)$ and $Q_\infty\cap\mathbb{F}_{q^6}(x)$. 

Now, we present several families of LCPs of AG Codes over the function field $F(\mathcal{X}_m)/\mathbb{F}_{q^6}$.


\begin{theorem}
    Let $n = q^6 + d(q-1)\big((m-1)q^3-1\big)-1$. If $V_F = \{\lambda,d^\prime\}$ and $d< \big\lfloor \frac{d(q-2)+d\lambda+d^\prime}{m} \big\rfloor$. Then 
    \begin{itemize}
        \item for $m\le s <(n-g+2)/(d(q-2)+d\lambda+d^\prime)$, the pair $(\mathcal{C}_{\mathcal{L}}(D,G),\mathcal{C}_{\mathcal{L}}(D,H))$ with parameters
        \begin{gather*}
            [n,n-(d(q-2)+d\lambda+d^\prime)s+1,\ge (d(q-2)+d\lambda+d^\prime)s-g]_{q^6}\\
            \text{ and }[n,(d(q-2)+d\lambda+d^\prime)s-1,\ge n-(d(q-2)+d\lambda+d^\prime)s+2-g]_{q^6}
        \end{gather*}
        respectively is an LCP. 
        \item for $m\le s <(n-g+1)/(d(q-2)+d\lambda+d^\prime-1)$, the pair $(\mathcal{C}_{\mathcal{L}}(D,G),\mathcal{C}_{\mathcal{L}}(D,H))$ with parameters
        \begin{gather*}
            [n,n-(d(q-2)+d\lambda+d^\prime-1)s,\ge (d(q-2)+d\lambda+d^\prime-1)s+1-g]_{q^6}\\
            \text{ and }[n,(d(q-2)+d\lambda+d^\prime-1)s,\ge n-(d(q-2)+d\lambda+d^\prime-1)s+1-g]_{q^6}
        \end{gather*}
        respectively is an LCP. 
    \end{itemize}
    Moreover, the LCPs above have the security parameter $d(\mathcal{C}_{\mathcal{L}}(D,G))$. 
\end{theorem}
\begin{proof}
    Through the above deduction, we know that the rational places in $\mathbb{P}_{\mathbb{F}_{q^6}(x)}\backslash\{Q_i\cap\mathbb{F}_{q^6}(x),Q_j^\prime\cap\mathbb{F}_{q^6}(x),Q_\infty\cap\mathbb{F}_{q^6}(x)\}$ split completely in the extension $F(\mathcal{X}_m)/\mathbb{F}_{q^6}(x)$. Let 
    \[
        D = \sum_{\substack{P\in\mathbb{P}_{F(\mathcal{X}_m)}\backslash\{Q_i,Q_j^\prime,Q_\infty\}\\  P \text{ is rational}}}P.
    \]
    Then we have $\deg D = q^6 + d(q-1)\big((m-1)q^3-1\big)-1 = n$. Moreover, by Lemma \ref{ceil_prop} and Equation (\ref{s_t_w=0}), if $d< \big\lfloor \frac{d(q-2)+d\lambda+d^\prime}{m} \big\rfloor$ holds, then we can conclude that $s_t\ge 0$\,(defined in Theorem~\ref{non-special_divisors_cons}) for $1\le t\le m-1$ which satisfies the assumption of Theorem \ref{non-special_divisors_cons}. Therefore, we can apply Theorems \ref{con_w=0_1} and \ref{con_w=0_2} to the Kummer extension $F(\mathcal{X}_m)/\mathbb{F}_{q^6}(x)$ and then we obtain our results. The proof is completed. 
\end{proof}


In the following, we provide some specific examples about LCPs of AG Codes over the function field $F(\mathcal{X}_m)/\mathbb{F}_{q^6}$.

\begin{example}\label{LCP_Xm_ex}
    {\rm (1)} Let $q=2^3$ and $m=19$. For any $d\mid q+1$, we derive that $d< \big\lfloor \frac{d(q-2)+d\lambda+d^\prime}{m} \big\rfloor$ holds. Keep the notations above and let $d = 1$, then we have $\lambda = d^\prime = 17$ and $V_F = \{ 17 \}$. Therefore, we obtain the following LCPs:
    \begin{itemize}
        \item For $19\le s <(n-61)/40$, the LCP $(\mathcal{C}_{\mathcal{L}}(D,G),\mathcal{C}_{\mathcal{L}}(D,H))$ with parameters
        \begin{gather*}
            [n,n-40s+1,\ge 40s-63]_{2^{18}}\text{ and }[n,40s-1,\ge n-40s-61]_{2^{18}}
        \end{gather*}
        has the security parameter $d(\mathcal{C}_{\mathcal{L}}(D,G))\ge 40s-63$.
        \item For $19\le s <(n-62)/39$, the pair $(\mathcal{C}_{\mathcal{L}}(D,G),\mathcal{C}_{\mathcal{L}}(D,H))$ with parameters
        \begin{gather*}
            [n,n-39s,\ge 39s-62]_{2^{18}}\text{ and }[n,39s,\ge n-39s-62]_{2^{18}}
        \end{gather*}
        has the security parameter $d(\mathcal{C}_{\mathcal{L}}(D,G))\ge 39s-62$.
    \end{itemize}

    {\rm (2)} Let $q = 2^6$ and $m = 37$. When $d \mid q+1$ and $d\neq 1$, we derive that $d<\big\lfloor \frac{d(q-2)+d\lambda+d^\prime}{m} \big\rfloor$. Therefore, if $V_F = \{\lambda,d^\prime\}$, then we can obtain some LCPs of AG Codes. However, when $d = 1$, we have $d=\big\lfloor \frac{d(q-2)+d\lambda+d^\prime}{m} \big\rfloor = 1$. Let $\lambda^\prime$ be an integer such that $1\le \lambda^\prime<m$ and $d(q-2)+d\lambda+d^\prime \equiv \lambda^\prime\bmod m$. Then we can derive that none of the following conditions hold: 
    \begin{itemize}
        \item The set $\big\{\big\lfloor\frac{l m}{\lambda}\big\rfloor\mid 1\le l\le \lambda-1\big\}$\,\big(or $\big\{\big\lfloor\frac{l m}{d^\prime}\big\rfloor\mid 1\le l\le d^\prime-1\big\}$\big) is a subset of $\big\{\big\lfloor\frac{l m}{\lambda^\prime}\big\rfloor\mid 1\le l\le \lambda^\prime-1\big\}$;
        \item $\big\{\big\lfloor\frac{l m}{\lambda}\big\rfloor\mid 1\le l\le \lambda-1\big\}\cap \big\{\big\lfloor\frac{l m}{d^\prime}\big\rfloor\mid 1\le l\le d^\prime-1\big\} = \varnothing$. 
    \end{itemize}
    It implies that there exists $1\le t_0\le m-1$ such that $s_{t_0}<0$. Hence, Remark \ref{con_w=0_rem}(1) could not be applied to the Kummer extension $F(\mathcal{X}_m)/\mathbb{F}_{q^6}(x)$. That is to say, our constructions could not be used in this situation. 
\end{example}

\subsection{LCPs of AG Codes over the Function Field $F(\mathcal{Y}_m)$} 

Since $d\mid q-1$ and $m\mid(q^r+1)/(q+1)$, it yields that $(d,m)=1$. Let $-d\equiv d^\prime \bmod m$ with $1\le d^\prime<m$. Applying the morphism $\varphi(x,y) = \big(U,  W\cdot U^{\frac{d+d^\prime}{m}}\big)$ to the curve $\mathcal{Y}_m$, then we obtain a birationally equivalent curve
\[
    \mathfrak{C}^\prime:y^m = x^{d^\prime}(1-x^{d(q+1)}) = - \prod_{i=1}^{d(q+1)}(x-\alpha_i)\cdot x^{d^\prime}, 
\]
where $\alpha_i\in \{ a\in\mathbb{F}_{q^2}\mid a^{d(q+1)} = 1 \}$. According to \cite[Cor. 4.5]{Algebraic_Geometry}, the function field of the curve $\mathfrak{C}^\prime$ is isomorphic to $F(\mathcal{Y}_m)$. Therefore, avoiding abuse of notations, we could rewrite the affine curve of $F(\mathcal{Y}_m)$ as follows:
\[
    \mathcal{Y}_m : y^m = - \prod_{i=1}^{d(q+1)}(x-\alpha_i)\cdot x^{d^\prime}. 
\]
The genus of $F(\mathcal{Y}_m)$ is $g=\frac{d(m-1)(q+1)}{2}$. There are totally $q^{2r}+2gq^r+1$ rational places in $F(\mathcal{Y}_m)/\mathbb{F}_{q^{2r}}$, since $F(\mathcal{Y}_m)$ is maximal over $\mathbb{F}_{q^{2r}}$. Since we have
\[
    (d(q+1)+d^\prime,m) = (d(q+1)-d,m) = (d\cdot q,m) = 1,
\]
there exists a unique place at infinity. Let $Q_i$ be the only rational place corresponding to the zero of $x-\alpha_i$, $Q_0$ the only rational place corresponding to the zero of $x$ and $Q_\infty$ the unique place at infinity. Moreover, the places which are ramified in the Kummer extension $F(\mathcal{Y}_m)/\mathbb{F}_{q^{2r}}(x)$ are exactly $Q_i\cap\mathbb{F}_{q^{2r}}(x)$, $Q_0\cap\mathbb{F}_{q^{2r}}(x)$ and $Q_\infty\cap\mathbb{F}_{q^{2r}}(x)$. 

Now, we present several families of LCPs of AG Codes over the function field $F(\mathcal{Y}_m)/\mathbb{F}_{q^{2r}}$.  

\begin{theorem}
    Let $n = q^{2r} + d(q+1)\big((m-1)q^r-1\big)-1$. If $V_F = \{d^\prime\}$, then the following assertions hold.
    \begin{itemize}
        \item For $m\le s <(n-g+2)/(d(q+1)+d^\prime)$, the pair $(\mathcal{C}_{\mathcal{L}}(D,G),\mathcal{C}_{\mathcal{L}}(D,H))$ with parameters
        \begin{gather*}
            [n,n-(d(q+1)+d^\prime)s+1,\ge (d(q+1)+d^\prime)s-g]_{q^{2r}}\\
            \text{ and }[n,(d(q+1)+d^\prime)s-1,\ge n-(d(q+1)+d^\prime)s+2-g]_{q^{2r}}
        \end{gather*}
        respectively is an LCP. 
        \item For $m\le s <(n-g+1)/(d(q+1)+d^\prime-1)$, the pair $(\mathcal{C}_{\mathcal{L}}(D,G),\mathcal{C}_{\mathcal{L}}(D,H))$ with parameters
        \begin{gather*}
            [n,n-(d(q+1)+d^\prime-1)s,\ge (d(q+1)+d^\prime-1)s+1-g]_{q^{2r}}\\
            \text{ and }[n,(d(q+1)+d^\prime-1)s,\ge n-(d(q+1)+d^\prime-1)s+1-g]_{q^{2r}}
        \end{gather*}
        respectively is an LCP. 
    \end{itemize}
    Moreover, the LCPs above have the security parameter $d(\mathcal{C}_{\mathcal{L}}(D,G))$. 
\end{theorem}
\begin{proof}
    Through the above deduction, we know that the rational places in $\mathbb{P}_{\mathbb{F}_{q^{2r}}(x)}\backslash\{Q_i\cap\mathbb{F}_{q^{2r}}(x),Q_0\cap\mathbb{F}_{q^{2r}}(x),Q_\infty\cap\mathbb{F}_{q^{2r}}(x)\}$ split completely in the extension $F(\mathcal{Y}_m)/\mathbb{F}_{q^{2r}}(x)$. Let 
    \[
        D = \sum_{\substack{P\in\mathbb{P}_{F(\mathcal{Y}_m)}\backslash\{Q_i,Q_0,Q_\infty\}\\  P \text{ is rational}}}P.
    \]
    Then we have $\deg D = q^{2r} + d(q+1)\big((m-1)q^r-1\big)-1 = n$. Furthermore, it consistently holds that $\big\lfloor \frac{d(q+1)+d^\prime}{m} \big\rfloor>0$. By Lemma \ref{ceil_prop} and Equation (\ref{s_t_w=0}), we can conclude that $s_t\ge 0$\,(defined in Theorem \ref{non-special_divisors_cons}) for $1\le t\le m-1$ which satisfies the assumption of Theorem \ref{non-special_divisors_cons}. Therefore, we can apply Theorems~\ref{con_w=0_1} and~\ref{con_w=0_2} to the Kummer extension $F(\mathcal{Y}_m)/\mathbb{F}_{q^{2r}}(x)$ and then we obtain our results. The proof is completed. 
\end{proof}

In the following, we provide some specific examples about LCPs of AG Codes over the function field $F(\mathcal{Y}_m)/\mathbb{F}_{q^{2r}}$.

\begin{example}\label{LCP_Ym_ex}
    {\rm (1)} Let $q=2$ and $r=3$. Then $d=1$, $m=3$ and 
    \[
        \mathcal{Y}_m/\mathbb{F}_{64} : y^3 = x^2\cdot(1-x^3) = - (x-\alpha_1)(x-\alpha_2)(x-\alpha_3)\cdot x^2, 
    \]
    where $\alpha_i\in \{ a\in\mathbb{F}_{4}\mid a^{3} = 1 \}$. Keep the notations above and we have $V_F = \{2\}$. Then we obtain the following LCPs:
    \begin{itemize}
        \item For $3\le s \le21$, the LCP $(\mathcal{C}_{\mathcal{L}}(D,G),\mathcal{C}_{\mathcal{L}}(D,H))$ with parameters
        \[
            [108,109-5s,\ge 5s-3]_{2^6}\text{ and }[108,5s-1,\ge 107-5s]_{2^6}
        \]
        has the security parameter $d(\mathcal{C}_{\mathcal{L}}(D,G))\ge 5s-3$.
        \item For $3\le s \le 26$, the LCP $(\mathcal{C}_{\mathcal{L}}(D,G),\mathcal{C}_{\mathcal{L}}(D,H))$ with parameters
        \[
            [108,108-4s,\ge 4s-2]_{2^6}\text{ and }[108,4s,\ge 106-4s]_{2^6}
        \]
        has the security parameter $d(\mathcal{C}_{\mathcal{L}}(D,G))\ge 4s-2$.
    \end{itemize}

    {\rm (2)} Let $q = 2$ and $r = 9$. Then $d = 1$. Let $m = 19\mid (q^r+1)/(q+1)$. Keep the notations above and we have $V_F = \{18\}$. Then we obtain the following LCPs:
    \begin{itemize}
        \item For $19\le s <(n-25)/21$, the LCP $(\mathcal{C}_{\mathcal{L}}(D,G),\mathcal{C}_{\mathcal{L}}(D,H))$ with parameters
        \begin{gather*}
            [n,n-21s+1,\ge 21s-27]_{2^{18}}\text{ and }[n,21s-1,\ge n-21s-25]_{2^{18}}
        \end{gather*}
        has the security parameter $d(\mathcal{C}_{\mathcal{L}}(D,G))\ge 21s-27$.
        \item For $19\le s <(n-26)/20$, the LCP $(\mathcal{C}_{\mathcal{L}}(D,G),\mathcal{C}_{\mathcal{L}}(D,H))$ with parameters
        \begin{gather*}
            [n,n-20s,\ge 20s-26]_{2^{18}}\text{ and }[n,20s,\ge n-20s-26]_{2^{18}}
        \end{gather*}
        has the security parameter $d(\mathcal{C}_{\mathcal{L}}(D,G))\ge 20s-26$.
    \end{itemize}
    It should be noted that the aforementioned two families of LCPs are different from those presented in Example \ref{LCP_Xm_ex}(1). 
\end{example}

\subsection{LCPs of AG Codes over Elliptic and Hyperelliptic Function Fields}

Let $q$ is a power of an odd prime number. We consider the function field $F_g/\mathbb{F}_q$ defined by the affine equation
\begin{equation}\label{hyperelliptic_g}
    \mathfrak{C}_g: y^2 = \prod_{i=1}^{2g+1}(x-\alpha_i),
\end{equation}
where $\alpha_i\in\mathbb{F}_q$ are pairwise distinct elements and $g\ge1$ is an integer. Let $Q_i$ be the only rational place corresponding to the zero of $x-\alpha_i$ and $Q_\infty$ the unique place at infinity. Moreover, the places which are ramified in the Kummer extension $F_g/\mathbb{F}_q(x)$ are exactly $Q_i\cap\mathbb{F}_q(x)$ and $Q_\infty\cap\mathbb{F}_q(x)$. By Section \ref{construction}, the function field $F_g/\mathbb{F}_q$ has genus $g$. As observed, $F_1/\mathbb{F}_q$ corresponds precisely to the elliptic function field, and $F_g/\mathbb{F}_q$ is the hyperelliptic function field when $g>1$. Let $N_g$ denote the number of rational places in $F_g/\mathbb{F}_q$. 

Now, we present several families of LCPs of AG Codes over the function field $F_g/\mathbb{F}_q$. 

\begin{theorem}\label{LCP_HEC}
    Let $n = N_g-2g-2$. Then the following assertions hold.
    \begin{itemize}
        \item For $2\le s <(n-g+2)/(2g+1)$, the pair $(\mathcal{C}_{\mathcal{L}}(D,G),\mathcal{C}_{\mathcal{L}}(D,H))$ with parameters
        \begin{gather*}
            [n,n-(2g+1)s+1,\ge (2g+1)s-g]_q\\
            \text{ and }[n,(2g+1)s-1,\ge n-(2g+1)s+2-g]_q
        \end{gather*}
        respectively is an LCP. 
        \item For $2\le s <(n-g+1)/2g$, the pair $(\mathcal{C}_{\mathcal{L}}(D,G),\mathcal{C}_{\mathcal{L}}(D,H))$ with parameters
        \begin{gather*}
            [n,n-2gs,\ge 2gs+1-g]_q\\
            \text{ and }[n,2gs,\ge n-2gs+1-g]_q
        \end{gather*}
        respectively is an LCP. 
        \item For $2\le s <(n-g+1)/(2g+1)$, the pair $(\mathcal{C}_{\mathcal{L}}(D,G),\mathcal{C}_{\mathcal{L}}(D,H))$ with parameters
        \begin{gather*}
            [n,n-(2g+1)s,\ge (2g+1)s+1-g]_q\\
            \text{ and }[n,(2g+1)s,\ge n-(2g+1)s+1-g]_q
        \end{gather*}
        respectively is an LCP. 
    \end{itemize}
    Moreover, the LCPs above have the security parameter $d(\mathcal{C}_{\mathcal{L}}(D,G))$. 
\end{theorem}
\begin{proof}
    Through the above deduction, we know that the rational places in $\mathbb{P}_{\mathbb{F}_q(x)}\backslash\{Q_i\cap\mathbb{F}_q(x),Q_\infty\cap\mathbb{F}_q(x)\}$ split completely in the extension $F_g/\mathbb{F}_q(x)$. Let 
    \[
        D = \sum_{\substack{P\in\mathbb{P}_{F_g}\backslash\{Q_i,Q_\infty\}\\  P \text{ is rational}}}P.
    \]
    Then we have $\deg D = N_g-2g-2 = n$. Subsequently, we can apply Theorems \ref{con_w=0_1}, \ref{con_w=0_2} and \ref{con_v=0} to the Kummer extension $F_g/\mathbb{F}_q(x)$ and then we obtain our results. The proof is completed. 
\end{proof}

For $F_1/\mathbb{F}_q$, i.e. the elliptic function field. It is well known that the set of rational places of $F_1/\mathbb{F}_q$ can be viewed as the abelian group $({\rm Jac}(\mathfrak{C}_1),\oplus)$ with $Q_\infty$ as the neutral element (see \cite{The_Arithmetic_of_Elliptic_Curves} or \cite[Prop. 6.1.7]{Algebraic_Function_Fields_and_Codes}). Suppose that $\mathcal{C}_{\mathcal{L}}(D,G)$ is an $[n,k,d]$ AG Code. Then $n-\deg G=n-k\le d\le n-k+1$ by (\ref{singleton_bound}). Moreover, $\mathcal{C}_{\mathcal{L}}(D,G)$ is MDS if and only if there does not exist $k$ distinct places $P_{i_1},\cdots,P_{i_k}\in{\rm supp}(D)$ such that $P_{i_1}+\cdots+P_{i_k}\sim G$. From \cite[Thm. 4.1]{Elliptic_Curves:_Number_Theory_and_Cryptography}, we know that
\[
    {\rm Jac}(\mathfrak{C}_1)\cong\mathbb{Z}/2^{\ell_1}m_1\mathbb{Z}\times\mathbb{Z}/2^{\ell_2}m_2\mathbb{Z},
\]
where $N_1 = 2^{\ell_1+\ell_2}m_1m_2$, $1\le\ell_1\le\ell_2$, $m_1\mid m_2$ and $(m_i,2)=1$ for $i=1,2$. That is to say, there exist rational places $R_1$, $R_2$ such that ${\rm ord}(R_1)=2^{\ell_1}m_1$, ${\rm ord}(R_2)=2^{\ell_2}m_2$ in ${\rm Jac}(\mathfrak{C}_1)$ and ${\rm Jac}(\mathfrak{C}_1)=\langle R_1, R_2\rangle$. Following this, we obtain several families of LCPs of MDS AG Codes over the function field $F_1/\mathbb{F}_q$. 

\begin{corollary}\label{LCP_EC}
    If $8\mid N_1$. Then 
    \begin{itemize}
        \item[(1)] for $2\le s <(N_1+2)/6$ and $2\mid s$, the pair $(\mathcal{C}_{\mathcal{L}}(D,G),\mathcal{C}_{\mathcal{L}}(D,H))$ with parameters
        \begin{gather*}
            [N_1/2,N_1/2-3s+1,3s]_q\text{ and }[N_1/2,3s-1,N_1/2-3s+2]_q
        \end{gather*}
        respectively is an LCP of MDS AG Codes with the security parameter $3s$. 
        \item[(2)] for $2\le s <N_1/6$ and $2\nmid s$, the pair $(\mathcal{C}_{\mathcal{L}}(D,G),\mathcal{C}_{\mathcal{L}}(D,H))$ with parameters
        \begin{gather*}
            [N_1/2,N_1/2-3s,3s+1]_q\text{ and }[N_1/2,3s,N_1/2-3s+1]_q
        \end{gather*}
        respectively is an LCP of MDS AG Codes with the security parameter $3s+1$. 
    \end{itemize}
    If $8\nmid N_1$, then
    \begin{itemize}
        \item[(3)] for $2\le s <(N_1+2)/6$ and $2\mid s$, the pair $(\mathcal{C}_{\mathcal{L}}(D,G),\mathcal{C}_{\mathcal{L}}(D,H))$ with parameters
        \begin{gather*}
            [N_1/2-2,N_1/2-3s-1,3s]_q\text{ and }[N_1/2-2,3s-1,N_1/2-3s]_q
        \end{gather*}
        respectively is an LCP of MDS AG Codes with the security parameter $3s$. 
        \item[(4)] for $2\le s <N_1/6$ and $2\nmid s$, the pair $(\mathcal{C}_{\mathcal{L}}(D,G),\mathcal{C}_{\mathcal{L}}(D,H))$ with parameters
        \begin{gather*}
            [N_1/2-2,N_1/2-3s-2,3s+1]_q\text{ and }[N_1/2-2,3s,N_1/2-3s-1]_q
        \end{gather*}
        respectively is an LCP of MDS AG Codes with the security parameter $3s+1$. 
    \end{itemize}
\end{corollary}
\begin{proof}
    (1) Since $8\mid N_1$, we have $\ell_2\ge 2$. Let $\mathcal{G} = \langle R_1, 2R_2\rangle$. Then $|\mathcal{G}| = N_1/2$ and $Q_i\in \mathcal{G}$ for $i=1,2,3,\infty$. Let $\mathcal{G}^\prime = R_2 \oplus \mathcal{G}$ be the coset of $\mathcal{G}$ and $D = \sum_{P\in \mathcal{G}^\prime} P$. We apply Theorem \ref{con_w=0_1} to $F_1/\mathbb{F}_q(x)$ and obtain the LCP. Let $k = N_1/2-3s+1$. Now we prove that $\mathcal{C}_{\mathcal{L}}(D,G)$ is an MDS AG Code. Suppose, on the contrary, that there exist $k$ distinct places $P_{i_1},\cdots,P_{i_k}\in{\rm supp}(D)$ such that $P_{i_1}+\cdots+ P_{i_k} \sim Q_1 + (N_1/2-3s)Q_\infty$. Therefore, by \cite[Prop. 6.1.7]{Algebraic_Function_Fields_and_Codes}, we have
    \begin{align*}
        P_{i_1}\oplus\cdots\oplus P_{i_k} = Q_1 &\Longleftrightarrow (R_2 \oplus P_{i_1}^\prime) \oplus\cdots\oplus (R_2 \oplus P_{i_k}^\prime) = Q_1 \qquad \big( P_{i_j}^\prime\in \mathcal{G} \big)\\
        & \Longleftrightarrow k R_2 \oplus P = Q_1 \qquad \big( P\in \mathcal{G} \big) \\
        & \Longleftrightarrow R_2 \in \mathcal{G}. \qquad\big( \text{A contradiction arises.} \big)
    \end{align*}
    The last equivalence arises from the fact that $k$ is an odd number. Hence, $\mathcal{C}_{\mathcal{L}}(D,G)$ is MDS and we obtain that the pair $(\mathcal{C}_{\mathcal{L}}(D,G),\mathcal{C}_{\mathcal{L}}(D,H))$ with parameters
    \begin{gather*}
        [N_1/2,N_1/2-3s+1,3s]_q\text{ and }[N_1/2,3s-1,N_1/2-3s+2]_q
    \end{gather*}
    respectively is an LCP with the security parameter $3s$. 

    (2) Similar to the proof of (1), we consider $\mathcal{G}^\prime = R_2 \oplus \mathcal{G}$, the coset of $\mathcal{G}= \langle R_1, 2R_2\rangle$. Let $D = \sum_{P\in \mathcal{G}^\prime} P$. We apply Theorem \ref{con_v=0} to $F_1/\mathbb{F}_q(x)$ and then obtain the result. 

    (3) Since $8\nmid N_1$, we have $\ell_1=\ell_2=1$. Without loss of generality, we may assume that $Q_1 = m_1 R_1$. Let $\mathcal{G} = \langle R_1, 2R_2\rangle$. Then  $Q_1,Q_\infty\in \mathcal{G}$ and $Q_2,Q_3\notin \mathcal{G}$. Subsequently, let $\mathcal{G}^\prime = R_2 \oplus \mathcal{G}$ be the coset of $\mathcal{G}$ and $D = \sum_{P\in \mathcal{G}^\prime\backslash\{Q_2,Q_3\}} P$. We apply Theorem \ref{con_w=0_1} to $F_1/\mathbb{F}_q(x)$ and obtain the LCP. Let $k = N_1/2-3s-1$. Now we prove that $\mathcal{C}_{\mathcal{L}}(D,G)$ is an MDS AG Code. Suppose, on the contrary, that there exist $k$ distinct places $P_{i_1},\cdots,P_{i_k}\in{\rm supp}(D)$ such that $P_{i_1}+\cdots+ P_{i_k} \sim Q_1 + (N_1/2-3s-2)Q_\infty$. Therefore, we have
    \begin{align*}
        P_{i_1}\oplus\cdots\oplus P_{i_k} = Q_1 &\Longleftrightarrow (R_2 \oplus P_{i_1}^\prime) \oplus\cdots\oplus (R_2 \oplus P_{i_k}^\prime) = Q_1 \qquad \big( P_{i_j}^\prime\in \mathcal{G} \big)\\
        & \Longleftrightarrow k R_2 \oplus P = Q_1 \qquad \big( P\in \mathcal{G} \big) \\
        & \Longleftrightarrow R_2 \in \mathcal{G}. \qquad\big( \text{A contradiction arises.} \big)
    \end{align*}
    The last equivalence arises from the fact that $k$ is an odd number. Hence, $\mathcal{C}_{\mathcal{L}}(D,G)$ is MDS and we obtain that the pair $(\mathcal{C}_{\mathcal{L}}(D,G),\mathcal{C}_{\mathcal{L}}(D,H))$ with parameters
    \begin{gather*}
        [N_1/2-2,N_1/2-3s-1,3s]_q\text{ and }[N_1/2-2,3s-1,N_1/2-3s]_q
    \end{gather*}
    respectively is an LCP with the security parameter $3s$. 

    (4) Similar to the proof of (3), we consider $\mathcal{G}^\prime = R_2 \oplus \mathcal{G}$, the coset of $\mathcal{G}= \langle R_1, 2R_2\rangle$. Let $D = \sum_{P\in \mathcal{G}^\prime\backslash\{Q_2,Q_3\}} P$. We apply Theorem \ref{con_v=0} to $F_1/\mathbb{F}_q(x)$ and then obtain the result. The proof is completed. 
\end{proof}

In the following, we provide some specific examples about LCPs of AG Codes over the elliptic function fields and hyperelliptic function fields.

\begin{example}\label{LCP_EC_ex}

    {\rm (1)} Let $p$ be an odd prime number with $p\equiv 2 \bmod 3$. By \cite[Lem. 17]{Optimal_Locally_Repairable_Codes_Via_Elliptic_Curves}, the elliptic curve $\mathfrak{C}_1$ over $\mathbb{F}_{p^2}$ defined by 
    \[
        y^2 = x^3+1,
    \]
    is maximal. Let $p = 3k+2$ and $k = 2l-1$. Then we derive that $6\mid p^2-1$ and
    \[
        y^2 = x^3+1 = (x-\alpha)(x-\alpha^3)(x-\alpha^5),
    \]
    where $\alpha$ is a $6$-th primitive unit root in $\mathbb{F}_{p^2}$. 
    
    When $2\mid l$, we have $8\mid p^2 + 2p + 1$. By Corollary~\ref{LCP_EC}(1) and (2), we can obtain the following LCPs of MDS AG Codes:
    \begin{itemize}
        \item For $2\le s <(p^2 + 2p + 3)/6$ and $2\mid s$, the LCP $(\mathcal{C}_{\mathcal{L}}(D,G),\mathcal{C}_{\mathcal{L}}(D,H))$ with parameters
        \begin{gather*}
            [(p^2 + 2p + 1)/2,(p^2 + 2p + 1)/2-3s+1,3s]_{p^2}\\
            \text{ and }[(p^2 + 2p + 1)/2,3s-1,(p^2 + 2p + 1)/2-3s+2]_{p^2}
        \end{gather*}
        has the security parameter $3s$. 
        \item For $2\le s <(p^2 + 2p + 1)/6$ and $2\nmid s$, the LCP  $(\mathcal{C}_{\mathcal{L}}(D,G),\mathcal{C}_{\mathcal{L}}(D,H))$ with parameters
        \begin{gather*}
            [(p^2 + 2p + 1)/2,(p^2 + 2p + 1)/2-3s,3s+1]_{p^2}\\
            \text{ and }[(p^2 + 2p + 1)/2,3s,(p^2 + 2p + 1)/2-3s+1]_{p^2}
        \end{gather*}
        has the security parameter $3s+1$. 
    \end{itemize}

    When $2\nmid l$, we have $8\nmid p^2 + 2p + 1$. By Corollary~\ref{LCP_EC}(3) and (4), we can obtain the following LCPs of MDS AG Codes:
    \begin{itemize}
        \item For $2\le s <(p^2 + 2p + 3)/6$ and $2\mid s$, the LCP $(\mathcal{C}_{\mathcal{L}}(D,G),\mathcal{C}_{\mathcal{L}}(D,H))$ with parameters
        \begin{gather*}
            [(p^2 + 2p + 1)/2-2,(p^2 + 2p + 1)/2-3s-1,3s]_{p^2}\\
            \text{ and }[(p^2 + 2p + 1)/2-2,3s-1,(p^2 + 2p + 1)/2-3s]_{p^2}
        \end{gather*}
        has the security parameter $3s$. 
        \item For $2\le s <(p^2 + 2p + 1)/6$ and $2\nmid s$, the LCP  $(\mathcal{C}_{\mathcal{L}}(D,G),\mathcal{C}_{\mathcal{L}}(D,H))$ with parameters
        \begin{gather*}
            [(p^2 + 2p + 1)/2-2,(p^2 + 2p + 1)/2-3s-2,3s+1]_{p^2}\\
            \text{ and }[(p^2 + 2p + 1)/2-2,3s,(p^2 + 2p + 1)/2-3s-1]_{p^2}
        \end{gather*}
        has the security parameter $3s+1$. 
    \end{itemize}

    {\rm (2)} Let $q$ be a power of an odd prime number with $q\equiv -1 \text{ or } 2g+1 \bmod 4g$ and $g\ge1$. By \cite[Thm. 1]{A_note_on_certain_maximal_hyperelliptic_curves}, the curve $\mathfrak{C}_g$ over $\mathbb{F}_{q^2}$ defined by 
    \[
        y^2 = x^{2g+1}+x = x\prod_{i = 1}^{2g}(x- \alpha^{2i-1}),
    \]
    is maximal where $\alpha$ is a $4g$-th primitive unit root in $\mathbb{F}_{q^2}$. By Theorem \ref{LCP_HEC}, we can obtain the following LCPs:
    \begin{itemize}
        \item For $2\le s <(q^2+(2q-3)g+1)/(2g+1)$, the LCP $(\mathcal{C}_{\mathcal{L}}(D,G),\mathcal{C}_{\mathcal{L}}(D,H))$ with parameters
        \begin{gather*}
            [q^2+2g(q-1)-1,q^2+2g(q-1-s)-s,\ge (2g+1)s-g]_{q^2}\\
            \text{ and }[q^2+2g(q-1)-1,(2g+1)s-1,\ge q^2+2g(q-1-s)-s+1-g]_{q^2}
        \end{gather*}
        has the security parameter $d(\mathcal{C}_{\mathcal{L}}(D,G))\ge (2g+1)s-g$. 
        \item For $2\le s <(q^2+(2q-3)g)/2g$, the LCP $(\mathcal{C}_{\mathcal{L}}(D,G),\mathcal{C}_{\mathcal{L}}(D,H))$ with parameters
        \begin{gather*}
            [q^2+2g(q-1)-1,q^2+2g(q-1-s)-1,\ge 2gs+1-g]_{q^2}\\
            \text{ and }[q^2+2g(q-1)-1,2gs,\ge q^2+2g(q-1-s)-g]_{q^2}
        \end{gather*}
        has the security parameter $d(\mathcal{C}_{\mathcal{L}}(D,G))\ge 2gs+1-g$.
        \item For $2\le s <(q^2+(2q-3)g)/(2g+1)$, the LCP $(\mathcal{C}_{\mathcal{L}}(D,G),\mathcal{C}_{\mathcal{L}}(D,H))$ with parameters
        \begin{gather*}
            [q^2+2g(q-1)-1,q^2+2g(q-1-s)-s-1,\ge (2g+1)s+1-g]_{q^2}\\
            \text{ and }[q^2+2g(q-1)-1,(2g+1)s,\ge q^2+2g(q-1-s)-s-g]_{q^2}
        \end{gather*}
        has the security parameter $d(\mathcal{C}_{\mathcal{L}}(D,G))\ge (2g+1)s+1-g$.
    \end{itemize}
\end{example}

Indeed, for the above curve $\mathfrak{C}_g$ in Example \ref{LCP_EC_ex}(2), if we modify the construction of the divisor $H$ in Theorems~\ref{con_w=0_1} and \ref{con_w=0_2}, then we can obtain some LCD codes.  

\begin{theorem}\label{LCD_HyperEC}
    Let $q$ be a power of an odd prime number with $q=2g+1$ and $g\ge1$. Let $F_g/\mathbb{F}_{q^2}$ be a maximal Kummer extension defined by the affine equation
    \[
        \mathfrak{C}_g: y^2 = x^q+x = \prod_{i=1}^{2g+1}(x-\alpha_i),
    \]
    where $\alpha_i\in\mathbb{F}_q$ are pairwise distinct elements in $\mathbb{F}_{q^2}$. Define the divisor $D = (y^t-1)_0$ with $t\ge4$ and $t\mid 4g$, then we can derive the following LCD codes:   
    \begin{itemize}
        \item[(1)] If we define 
        \[
            G = \sum_{i=1}^g Q_i + (4g+2)Q_\infty,
        \]
        then $\mathcal{C}_{\mathcal{L}}(D,G)$ is a $[t(2g+1),4g+3,\ge t(2g+1)-5g-2]_{q^2}$ LCD code. 
        \item[(2)] If we define 
        \[
            G = \sum_{i=1}^g Q_i + (t-2)Q_{2g+1} + (4g+1)Q_\infty,
        \]
        then $\mathcal{C}_{\mathcal{L}}(D,G)$ is a $[t(2g+1),4g+t,\ge 2tg-5g+1]_{q^2}$ LCD code. 
    \end{itemize}
\end{theorem}
\begin{proof}
    By \cite[Prop. 3.7.10]{Algebraic_Function_Fields_and_Codes}, all rational places $P_\alpha$ in $\mathbb{P}_{\mathbb{F}_{q^2}(y)}$, except for the pole of $y$, are unramified in the extension $F_g/\mathbb{F}_{q^2}(y)$ and 
    \[
        d(Q_\infty|P_\infty) = 3(q-1) = 6g,
    \]
    where $P_\alpha$ corresponds to $y-\alpha$ and $P_\infty$ corresponds to the pole of $y$. Therefore we have
    \[
        {\rm Diff}(F_g/\mathbb{F}_{q^2}(y)) = 6g Q_\infty, 
    \]
    and 
    \[
        (dy) = -2(y)_\infty + {\rm Diff}(F_g/\mathbb{F}_{q^2}(y)) = (2g-2) Q_\infty.
    \]
    Since $x^q+x$ is the trace map from $\mathbb{F}_{q^2}$ to $\mathbb{F}_{q}$ for every $\alpha\in\mathbb{F}_{q^2}$ with $\alpha^{4g}=\alpha^{2(q-1)}=1$, we can conclude that $x^q+x = \alpha^2$ has $q$ distinct roots i.e. $P_\alpha$ splits completely in $F_g/\mathbb{F}_{q^2}(y)$. Thus, $\deg D = t(2g+1)$. 

    {\rm (1)} Let 
    \[
        H = \sum_{i=1}^g (t-2)Q_i + \sum_{i=g+1}^{2g+1} (t-1)Q_i - 3Q_\infty.
    \]
    Then we have
    \begin{align*}
        G+H-D & = \sum_{i=1}^{2g+1} (t-1)Q_i + (4g-1) Q_\infty - (y^t-1)_0 \\
        & = \sum_{i=1}^{2g+1} (t-1)Q_i - (2g+1)(t-1) Q_\infty + (2g-2)Q_\infty - (y^t-1)\\
        & = (y^{t-1}) + (dy) - (y^t-1)\\
        & = \bigg(\frac{d(y^t-1)}{y^t-1}\bigg)
    \end{align*}
    where $v_{P}(y^t-1)=1$ for all $P\in{\rm supp}(D)$. By Lemma \ref{dual_AG Code}, we derive that $\mathcal{C}_{\mathcal{L}}(D,G)^{\perp} = \mathcal{C}_{\mathcal{L}}(D,H)$. On the other hand, we know that $(\mathcal{C}_{\mathcal{L}}(D,G),\mathcal{C}_{\mathcal{L}}(D,H))$ is an LCP by Theorem \ref{con_w=0_1}. Hence, $\mathcal{C}_{\mathcal{L}}(D,G)$ is an LCD code with parameters $[t(2g+1),4g+3,\ge t(2g+1)-5g-2]_{q^2}$. 

    {\rm (2)} Let 
    \[
        H = \sum_{i=1}^g (t-2)Q_i + \sum_{i=g+1}^{2g} (t-1)Q_i + Q_{2g+1} - 2Q_\infty.
    \]
    Similar to the proof of (1), we have $\mathcal{C}_{\mathcal{L}}(D,G)^{\perp} = \mathcal{C}_{\mathcal{L}}(D,H)$. It follows that $\mathcal{C}_{\mathcal{L}}(D,G)$ is an LCD code. The proof is completed. 
\end{proof}

In the following, we provide a specific example of LCD codes over hyperelliptic function field with genus $g=2$.

\begin{example}\label{LCD_HyperEC_ex}
    Let $q = 5$ and $g = 2$. Consider the function field $F_2/\mathbb{F}_{5^2}$ defined by
    \[
        \mathfrak{C}_2: y^2 = x^5+x. 
    \]
    Let $t = 4g =8$. Then we obtain 
    \begin{itemize}
        \item[(1)] a $[40,11,\ge 28]_{5^2}$ LCD code $\mathcal{C}_{\mathcal{L}}(D,G)$ with $D = (y^8-1)_0$ and $G = Q_1+Q_2 + 10 Q_\infty$, 
        \item[(2)] a $[40,16,\ge 23]_{5^2}$ LCD code $\mathcal{C}_{\mathcal{L}}(D,G)$ with $D = (y^8-1)_0$ and $G = Q_1+Q_2 + 6Q_5 + 9 Q_\infty$. 
    \end{itemize}
    In addition, the above LCD codes have the best-known parameters according to \cite{Mint}. 
\end{example}

\subsection{LCPs of AG Codes over the other Function Fields}

In the following, we will consider the Kummer extension with $w\neq0$ defined in \cite{kawakitaKummerCurvesTheir2003a} and present a family of LCPs of AG Codes. Let $F/\mathbb{F}_q$ denote the function field defined by 
\begin{equation}\label{other_F}
    y^m = x \cdot(x-1)^{m^\prime}\cdot\prod_{i = 1}^{m-m^\prime} (x-\alpha_i),
\end{equation}
where $\alpha_i\in\mathbb{F}_q\backslash\{0,1\}$ are pairwise distinct elements, $m\mid q-1$ and $m^\prime\mid m$. By Section \ref{construction}, the function field $F/\mathbb{F}_q$ has genus $g = \frac{m(m-m^\prime)}{2}$. Let $Q_i$ be the only rational place corresponding to the zero of $x-\alpha_i$, $Q_0$ the only rational place corresponding to the zero of $x$ and $Q_\infty$ the unique place at infinity. Let $N$ denote the number of rational places in $F/\mathbb{F}_q$. 

\begin{theorem}\label{LCP_w_neq0}
    Let $f(x) = \prod_{i = 1}^{m-m^\prime} (x-\alpha_i)$ and 
    \[
        n = \begin{cases}
            N - m - 2, &\text{if $f(1)$ is an $m^\prime$-th power in $\mathbb{F}_q$},\\
            N - m + m^\prime - 2 , &\text{otherwise}.
        \end{cases}
    \]Then for $m\le s <(n-g+1)/(m+1)$, the pair $(\mathcal{C}_{\mathcal{L}}(D,G),\mathcal{C}_{\mathcal{L}}(D,H))$ with parameters
        \begin{gather*}
            [n,n-(m+1)s,\ge (m+1)s+1-g]_q\\
            \text{ and }[n,(m+1)s,\ge n-(m+1)s+1-g]_q
        \end{gather*}
    respectively is an LCP with the security parameter $d(\mathcal{C}_{\mathcal{L}}(D,G))$. 
\end{theorem}
\begin{proof}
    Keep the notations in Theorem \ref{con_v=0}. Let 
    \[
        R = \sum_{Q\in\mathbb{P}_F,Q\mid P_1} Q,
    \]
    where $P_1$ represents the zero of $x-1$ in $\mathbb{P}_{\mathbb{F}_q(x)}$. By \cite[Prop. 1.4(ii)]{kawakitaKummerCurvesTheir2003a}, all of the places lying over $P_1$ are rational if and only if $f(1)$ is an $m^\prime$-th power in $\mathbb{F}_q$. Therefore, if we can set 
    \[
        D = \begin{cases}
            \sum_{\substack{P\in\mathbb{P}_{F}\backslash\{Q_0,Q_i,Q_\infty,Q\}\\  P \text{ is rational}}}P, &\text{if $f(1)$ is an $m^\prime$-th power in $\mathbb{F}_q$},\\
            \sum_{\substack{P\in\mathbb{P}_{F}\backslash\{Q_0,Q_i,Q_\infty\}\\  P \text{ is rational}}}P , &\text{otherwise}.
        \end{cases}
    \]
    So we have $\deg D = n$. By Lemma \ref{ceil_prop} and Equation (\ref{s_t_v=0}), we can conclude that $s_t\ge 0$\,(defined in Theorem \ref{non-special_divisors_cons}) for $1\le t\le m-1$ which satisfies the assumption of Theorem \ref{non-special_divisors_cons}. Hence, we can apply Theorem \ref{con_v=0} to the Kummer extension $F/\mathbb{F}_q(x)$ and then we obtain our result. The proof is completed. 
\end{proof}

At last, we provide a specific example about LCP of AG Codes over the function field $F/\mathbb{F}_q$ defined by Equation (\ref{other_F}). 

\begin{example}\label{LCP_w_neq0_ex}
    Let $q = 11^2$ and $\mathbb{F}_q = \mathbb{F}_{11}(u)$. Let $F/\mathbb{F}_q$ denote the function field defined by 
    \[
        y^4 = x \cdot(x-1)^2\cdot(x-u-5)\cdot(x - 10u-9). 
    \]
    By Sagemath, we have $N = 210$. That is to say, $F/\mathbb{F}_q$ is a maximal function field. Moreover, we have $f(1) = (1-u-5)\cdot(1 - 10u-9) = 1$. It is obviously that $f(1)$ is a $2$-th power in $\mathbb{F}_q$. Hence by Theorem~\ref{LCP_w_neq0}, for $4\le s \le40$, the LCP $(\mathcal{C}_{\mathcal{L}}(D,G),\mathcal{C}_{\mathcal{L}}(D,H))$ with parameters
        \begin{gather*}
            [204,204-5s,\ge 5s-3]_{11^2}\text{ and }[204,5s,\ge 201-5s]_{11^2}
        \end{gather*}
    has the security parameter $d(\mathcal{C}_{\mathcal{L}}(D,G))\ge 5s-3$. 
\end{example}

\section{Conclusion}\label{conclusion}

In this paper, we introduce several methods for constructing LCPs of AG Codes via Kummer extensions. The key aspect of our work is that we determine explicit non-special divisors of degree $g$ and $g-1$ on Kummer extensions with specific properties. Consequently, we obtain many families of LCPs from subcovers of the BM curve, elliptic function fields, hyperelliptic function fields and other function fields. In addition, we obtain two families of LCD codes from certain maximal elliptic function fields and hyperelliptic function fields. Our results have significantly enriched the construction of LCPs of AG Codes. Moreover, we explicitly derive some families of LCPs of MDS AG Codes from elliptic function fields in Corollary~\ref{LCP_EC}, which demonstrates that their security parameters are optimal.

\section*{Acknowledgement}

The work was supported by Guangdong Basic and Applied Basic Research Foundation~ (Grant No. 2025A1515011764), the National Natural Science Foundation of China (No. 12441107),  Guangdong Major Project of Basic and Applied Basic Research (No. 2019B030302008) and Guangdong Provincial Key Laboratory of Information Security Technology (No. 2023B1212060026).

\bibliographystyle{ieeetr}
\bibliography{reference}

\end{document}